\documentclass{amsart}
\usepackage[margin=3.18cm, top=2.54cm, bottom=2.54cm]{geometry}

\allowdisplaybreaks

\usepackage{caption}
\usepackage[leftcaption]{sidecap}

\usepackage{enumitem}
\setlist[enumerate,1]{label=\textbf{\textup{(\roman*)}}}
\setlist{
  leftmargin=1cm,
}

\usepackage[
  backend=biber,
  style=alphabetic,
  giveninits=true, 
  backref=true
]{biblatex}
\usepackage{lmodern}
\usepackage{anyfontsize}
    
\usepackage{multirow}
\usepackage{hhline}

\usepackage{tabularray}
\UseTblrLibrary{booktabs}

\usepackage{adjustbox}
\usepackage{tcolorbox}

\newtcolorbox{standout}{
  colback=gray!15,
  boxrule=0pt,
  left=.3cm,
  right=.3cm,
  top=.18cm,
  bottom=.18cm,
  boxsep=0pt
}

\usepackage{mathtools}
\usepackage{amssymb}
\usepackage{amsthm}
\usepackage{xfrac}
\usepackage{stmaryrd} 
\usepackage{scalerel} 
\usepackage{stackengine} 

 \newcommand{\bracket}[3]{%
  \stretchleftright
    {#1}
    {%
      \ensurestackMath{\addstackgap[1pt]{#2}}%
      \vrule width 0pt depth 2pt height 0pt
    }
    {#3}%
} 
\newcommand{\scaledbracket}[3]{%
  \ThisStyle{%
    \stretchleftright
      {#1}
      {
        \ensurestackMath{\addstackgap[1pt]{\SavedStyle #2}}%
        \vrule width 0pt depth 1.5pt height 0pt
      }
      {#3}%
  }%
}
\newcommand{\bracketmid}[4]{%
  \stretchleftright{#1}{%
    \ensurestackMath{%
      \addstackgap[2pt]{#2}%
      \,\stretchrel*{|}{\addstackgap[2pt]{#2#3}}\,%
      \addstackgap[2pt]{#3}%
    }%
  }{#4}%
}

\theoremstyle{plain}
\newtheorem{theorem}{Theorem}[section]
\newtheorem{lemma}[theorem]{Lemma}
\newtheorem{proposition}[theorem]{Proposition}

\theoremstyle{definition}

\newtheorem{example}[theorem]{Example}

\theoremstyle{remark}
\newtheorem{remark}[theorem]{Remark}
\usepackage[
  colorlinks=true,
  linkcolor=darkgreen,
  citecolor=darkgreen,
  urlcolor=darkgreen
]{hyperref}
\usepackage{xurl} 

\usepackage{cleveref}
\crefname{equation}{}{}
\crefname{section}{\S}{\S\S}
\crefname{subsection}{\S}{\S\S}
\crefname{subsubsection}{\S}{\S\S}
\crefname{definition}{Def.}{Defs.}
\crefname{theorem}{Thm.}{Thms.}
\crefname{corollary}{Cor.}{Cors.}
\crefname{lemma}{Lem.}{Lems.}
\crefname{proposition}{Prop.}{Props.}
\crefname{remark}{Rem.}{Rems.}
\crefname{notation}{Ntn.}{Ntns.}
\crefname{fact}{Fact}{Fact}
\crefname{example}{Ex.}{Exs.}
\crefname{figure}{Fig.}{Figs.}
\crefname{table}{Tab.}{Tabs.}
\crefname{footnote}{ftn.}{ftns.}
\Crefname{footnote}{Ftn.}{Ftns.}

\usepackage{tikz}
\usetikzlibrary{
  cd,
  calc,
  arrows.meta,
  backgrounds,
  decorations,
  decorations.pathmorphing, 
}

\tikzcdset{
  arrow style=tikz, 
  diagrams={
    trim left=
    {([xshift=5pt]current bounding box.west)},
    trim right=
    {([xshift=-5pt]current bounding box.east)},
    baseline=
    {([yshift=-2pt]current bounding box.center)},
    >={
      Computer Modern Rightarrow[
        length=4pt, width=4pt
      ]
    },
    hook/.style={
      {Hooks[right, length=2pt, width=8pt]}->
    },
    hook'/.style={
      {Hooks[left, length=2pt, width=8pt]}->
    },
    shorten=0pt,
  }
}

\definecolor{darkblue}{rgb}{0.05,0.25,0.65}
\definecolor{darkgreen}{RGB}{20,140,10}
\definecolor{lightgray}{rgb}{0.9,0.9,0.9}
\definecolor{darkorange}{RGB}{200,100,5}
\definecolor{darkyellow}{rgb}{.91,.91,0}
\definecolor{lightolive}{RGB}{225, 220, 185}

\let\originalsslash\sslash
\renewcommand{\sslash}{\mathord{\originalsslash}}

\makeatletter
\newcommand{\cpt}{\mathpalette\cpt@inner\relax}
\newcommand{\cpt@inner}[2]{%
  \scalebox{0.5}[0.9]{$#1\cup$}
  #1\{\infty\}
}
\makeatother

\tikzset{
  snake left/.style={
    rounded corners,
    to path={
      let \p1 = (\tikztostart.east),
          \p2 = (\tikztotarget.west),
          \p3 = ($(\p1)!0.5!(\p2)$),
          \n1 = {8pt} 
      in
      (\p1)
      -- (\x1 + \n1, \y1)
      -- (\x1 + \n1, \y3)
      -- (\x2 - \n1, \y3) \tikztonodes
      -- (\x2 - \n1, \y2)
      -- (\p2)
    }
  }
}

\tikzset{
  uphordown/.style={
    rounded corners,
    to path={
      let \p1 = (\tikztostart.north),
          \p2 = (\tikztotarget.north),
          \n1 = {max(\y1,\y2) + 8pt}
      in
      (\p1)
      -- (\x1, \n1)
      -- (\x2, \n1) \tikztonodes 
      -- (\p2)
    }
  }
}

\tikzset{
  downhorup/.style={
    rounded corners,
    to path={
      let \p1 = (\tikztostart.south),
          \p2 = (\tikztotarget.south),
          \n1 = {min(\y1,\y2) - 8pt}
      in
      (\p1)
      -- (\x1, \n1)
      -- (\x2, \n1) \tikztonodes 
      -- (\p2)
    }
  }
}

\tikzset{
  rightvertleft/.style={
    rounded corners,
    to path={
      let \p1 = (\tikztostart.east),
          \p2 = (\tikztotarget.east),
          \n1 = {max(\x1,\x2) + 8pt}
      in
      (\p1)
      -- (\n1, \y1)
      -- (\n1, \y2) \tikztonodes 
      -- (\p2)
    }
  }
}

\tikzset{
  leftvertright/.style={
    rounded corners,
    to path={
      let \p1 = (\tikztostart.west),
          \p2 = (\tikztotarget.west),
          \n1 = {min(\x1,\x2) - 8pt}
      in
      (\p1)
      -- (\n1, \y1)
      -- (\n1, \y2) \tikztonodes 
      -- (\p2)
    }
  }
}

\newcommand{\inlinetikzcd}[1]{\begin{tikzcd}[sep=small, ampersand replacement=\&]#1\end{tikzcd}}

\newcommand{\acts}{%
  \hspace{1.3pt}%
  \raisebox{1.2pt}{%
    \rotatebox[origin=c]{90}{$%
      \curvearrowright%
    $}%
  }%
  \hspace{.7pt}%
}

\renewcommand{\setminus}{-}

\newcommand{\backsslash}{\backslash\mspace{-6.5mu}\backslash}

\newcommand{\defneq}{\equiv}

\newcommand{\HilbertSpace}{%
  \mathcal{H}%
}

\newcommand{\fusion}{\widetilde\otimes}

\newcommand{\FusionMultiplicity}{R}

\begin{document}

\author{Hisham Sati}
\address{Center for Quantum \& Topological Systems, New York Univerity Abu Dhabi 
\& The Courant Institute for Mathematical Sciences, New York University}
\email{hsati@nyu.edu}

\author{Urs Schreiber}
\address{Center for Quantum \& Topological Systems, New York Univerity Abu Dhabi}
\email{us13@nyu.edu}

\thanks{This research was supported by {\it Tamkeen UAE} under the {\it Abu Dhabi Research Institute} grant {\tt CG008}.}

\title{Drinfeld Center as Quantum State Monodromy over Bloch Hamiltonians around Defects}

\begin{abstract}
The Drinfeld center fusion category $\mathcal{Z}\bracket({\mathrm{Vec}_G})$ famously models anyons in certain lattice models. Here we demonstrate how its fusion rules may also describe topological order in fractional topological insulator materials, in the vicinity of point defects in the Brillouin zone.

Concretely, we prove that $\mathcal{Z}\bracket({\mathrm{Vec}_G})$ reflects, locally over a punctured disk in the Brillouin zone, the monodromy (topological order) of gapped quantum states over the parameter space of Bloch Hamiltonians whose classifying space has fundamental group $G$.
\end{abstract}

\maketitle

\tableofcontents

\newpage

\section{Introduction}
\label{Introduction}

The \emph{Drinfeld center}  $\mathcal{Z}\bracket({\mathrm{Vec}_G})$
of the category of $G$-graded vector spaces
(recalled in \cref{OnTheDrinfeldCenter}) 
is a basic example of a \emph{fusion category} (cf. \cite{ENO2005})
\footnote{
  By the finiteness condition on the set of simple objects of a fusion category, $\mathcal{Z}\bracket({\mathrm{Vec}_G})$ is fusion when $G$ is a finite group, as is well-known. While our results all hold in the generality where $G$ may be non-finite, we will generally refer to $\mathcal{Z}\bracket({\mathrm{Vec}_G})$ as a fusion category, in order not to overburden the terminology.
}
and as such has received much attention (originating with \cite{Kitaev2003}, cf. \cite{BolsVadnerkar2025}) as a category of \emph{anyon} species (cf. \cite{Goldin2023}) in lattice models (\parencites[\S 4]{Kitaev2003}[\S E]{Kitaev2006}, cf. \cite[\S 29]{Simon2023}).

However, while there are many \emph{theoretical} models for anyons like this, a single anyon model currently stands out as being consistently observed in actual experiments, in recent years (starting with \cite{Bartolomei2020,Nakamura2020}, review in \cite{FeldmanHalperin2021}, most recent confirmations in \cite{Veillon2024,Ghosh2025}): namely \emph{quasi-holes} in fractional quantum Hall (FQH) liquids (cf. \cite{Stormer99,PapicBalram2024}). 

Yet more recently, the ``anomalous'' version (FQAH) of the fractional quantum Hall effect has been experimentally observed in crystalline quantum materials known as \emph{fractional Chern insulators} (FCI, cf. \cite{Chang2023,zhao2025exploring}). This is potentially of great practical relevance for future applications, notably to topological quantum computing hardware (cf. \cite{Nayak2008,SatiValera2025}), since the FQAH effect is seen under much more practical conditions (requiring less extreme cooling and no extreme magnetic field).
Therefore, an acutely relevant open question, both experimentally and theoretically, is the potential nature of anyonic topological order in such FCI materials (cf. \cite{SS25-FQAH,SS25-Crys}).

However, the topological phases of such crystalline insulators are naturally described not by lattice models, but by the homotopy of their \emph{Bloch Hamiltonians} (cf. \parencites[\S XIII.16]{ReedSimon1978}{Sergeev2023}): These are continuous maps from the \emph{Brillouin zone} $\Sigma^2$ of electron momenta in the crystal (cf. \cite[\S 2.1]{Thiang2025}), to a (classifying) space $\mathcal{A}$ of admissible Hamiltonians modeling the internal degrees of freedom of electrons of fixed quasi-momentum. 

Moreover (\cite[\S I-II]{SS25-Crys}, recalled in \cref{OnTopologicalOrderAndParameterMonodromy} below), since the single-electron Bloch Hamiltonian serves as the external parameter of crystal couplings on which the interacting electron quantum ground states must adiabatically depend, any topological order in these systems ought to manifest as representations of the fundamental groups in the space of Bloch Hamiltonians $H$.

Our main result here (in \cref{LocalParameterMonodromyViaCenter}) is a proof that, \emph{locally around point defects} in the Brillouin torus (such as around band nodes), this \emph{parameter monodromy} over spaces of Bloch Hamiltonians is faithfully reflected by the simple objects and their fusion rules in $\mathcal{Z}\bracket({\mathrm{Vec}_G})$.



\section{Topological Order and Parameter Monodromy}
\label{OnTopologicalOrderAndParameterMonodromy}

To start with, we recall, with \cite[\S I-II]{SS25-Crys}, that \emph{topological order} of quantum materials \cite{Wen1991,Wen1995}, exhibiting anyonic quantum states, is generally about nontrivial \emph{monodromy} of gapped quantum ground states with respect to adiabatic transport along paths of external parameters (\emph{local systems} of ground state Hilbert spaces over parameter space, cf. \parencites[Lit. 2.22 \& \S 3]{MySS2024}[\S 2.2]{SS26-EoS}). Then we specialize this to the situation of interest here, where the parameter space is that of Bloch Hamiltonians in the vicinity of a defect in momentum space.

\subsection{Classifying Spaces of Bloch Hamiltonians}
\label{OnClassifyingSpacesOfBlochHamiltonians}

We consider a path-connected topological space $\mathcal{A}$ admitting the structure of a CW-complex, to be regarded as a classifying space for Bloch Hamiltonians. Typical examples, classifying fragile topological phases \cite{Bouhon2020} of sequences of groups of gapped electron bands, are flag manifolds of the form
\begin{equation}
  \mathcal{A}
  \sim
  U\big/\bracket({
    U_1 \times \cdots \times U_k
  })
  \mathrlap{\,,}
\end{equation}
where $U \subset \mathrm{U}(\HilbertSpace)$ is a subgroup of unitary transformations (of the single electron's internal quantum states) preserved by the material's quantum symmetries, and the $U_i \subset U$ are the disjoint symmetry subgroups under which the gapped groups of bands transform.

\begin{example}
Consider the simple but important case of 2-band (fractional) Chern insulators (cf. \cite[\S A.1]{Chang2023}, which assumes that a single valence band and a single conduction band are accessible to the system under deformations, subject to no other symmetry constraints). Here we have that the classifying space is homotopy equivalent \eqref{HomotopyEquivalenceRelation} to the 2-sphere (cf. \parencites[(8.3-4)]{Sergeev2023}[(4)]{SS25-FQAH}):
\begin{equation}
  \label{S2Classifying2BandChernInsulators}
  \mathcal{A}
  \sim
  \mathrm{U}(2)\big/
  \bracket({
    \mathrm{U}(1)
    \times
    \mathrm{U}(1)
  })
  \simeq
  \mathbb{C}P^1
  \simeq
  S^2
  \mathrlap{.}
\end{equation}
\end{example}
Simplistic as this classifying space may appear, it plays a remarkable role in witnessing FQH-type anyons in the momentum space of FQAH systems, see \cref{HeisenbergGroupsAsPi1} below. 

But for the purposes of the present article,  we will instead be focused on examples of classifying spaces that have trivial $\pi_2$ but interesting $\pi_1$, such as the following:
\begin{example} 
  \label[example]{PTSymmetric3BandSystems}
  The classifying space of Bloch Hamiltonians for a PT-symmetric crystalline system with 3 gapped bands (considered in \parencites{WSB2019}[pp. 11]{Bouhon2023}) is:
  \begin{equation}
    \mathcal{A}
    \sim
    \mathrm{O}(3)\big/\bracket({
      \mathrm{O}(1)^3
    })
    \simeq
    \mathrm{SO}(3)
    \big/D_2
  \end{equation}
  This space has trivial $\pi_2$, and its fundamental group is the \emph{quaternion group}:
  \begin{equation}
    \pi_1(\mathcal{A})
    \simeq
    \bracket\{{
      \pm 1,
      \pm \mathrm{i},
      \pm \mathrm{j},
      \pm \mathrm{k}
    }\}
    \subset 
    S(\mathbb{H})
    \,,
    \;\;\;\;\;\;
    \pi_2(\mathcal{A}) = 0
    \mathrlap{\,.}
  \end{equation}
\end{example}

\subsection{Homotopy of spaces of Bloch Hamiltonians}
\label{OnHomotopyOfSpacesOfBlochHamiltonians}

With a Bloch Hamiltonian classifying space $\mathcal{A}$ given (\cref{OnClassifyingSpacesOfBlochHamiltonians}),
and denoting by $\Sigma^2$ a domain of crystal momenta, the space of Bloch Hamiltonians on the domain with the given properties is homotopy equivalent \cref{HomotopyEquivalenceRelation} to the following mapping space:
\begin{equation}
  \big\{
    \text{Bloch Hamiltonians on $\Sigma^2$ classified by $\mathcal{A}$}
  \big\}
  =
  \mathrm{Map}\bracket({\Sigma^2, \mathcal{A}})
  \defneq
  \big\{
  \begin{tikzcd}[sep=15pt]
    \Sigma^2
    \ar[r, dashed]
    &
    \mathcal{A}
  \end{tikzcd}
  \big\}
  \mathrlap{\,.}
\end{equation}
Since Bloch Hamiltonians encode the crystal/potential structure seen by valence electrons, this may be regarded as the space of classical external parameters for the multi-electron quantum system.

But this means \cite[\S I-II]{SS25-Crys}, by the \emph{quantum adiabatic theorem} (cf. \cite{RigolinOrtiz2012}), that gapped quantum ground states of interacting electrons form \emph{local systems} (cf. \parencites[Lit. 2.22]{MySS2024}[\S 2.2]{SS26-EoS}) of Hilbert spaces over this external parameter space, hence here form representations of the fundamental groups $\pi_1$ of this mapping space (cf. Fig. \ref{ParameterMonodromySchematics}):
\begin{enumerate}
  \item 
  $[p] \in \pi_0\bracket({\mathrm{Map}\bracket({{\Sigma^2, \mathcal{A}}})})$
  is a \emph{topological phase},

  \item
  $\pi_1\bracket({\mathrm{Map}\bracket({\Sigma^2,\mathcal{A}}), p})$
  is the \emph{parameter monodromy}
  in that phase, signifyin \emph{topological order} (cf. Fig. \ref{BlochMonodromySchematics}). 
\end{enumerate}

\begin{SCfigure}[2][htb]
\centering
\caption{
  \label{ParameterMonodromySchematics}
  The \emph{quantum adiabatic theorem} entails that gapped ground states $\HilbertSpace$ undergo unitary transformations $U_\gamma$ when classical parameters $p$ move along paths $\gamma$ in parameter space. If these $U_\gamma$ depend only on the homotopy class of $\gamma$ (relative endpoints) then nontrivial such transformations signify \emph{topological order} of the quantum phase. Mathematically this means that the Hilbert spaces $\HilbertSpace_p$ constitute a \emph{local system} or \emph{flat bundle} over parameter space, equivalently a representation of the \emph{fundamental group} of closed paths (loops) $\ell$ at any base point $p_0$.
}
\centering
\adjustbox{
  rndfbox=4pt
}{
$
\begin{tikzcd}[
  decoration=snake,
  row sep=20pt
]
  &
  \HilbertSpace_{p_2}
  \ar[
    dr,
    shorten=-2pt,
    "{ U_{\gamma_{{}_{23}}} }"{sloped}
  ]
  \\[-11pt]
  \HilbertSpace_{\mathrlap{p_{{}_{1}}}}
  \ar[
    rr,
    shorten >=-1pt,
    black,
    "{ U_{\gamma_{{}_{13}}} }"
  ]
  \ar[
    ur,
    shorten >=-2pt,
    "{
      U_{\gamma_{{}_{12}}}
    }"{sloped}
  ]
  &&
  \HilbertSpace_{p_{{}_3}}
  &[-13pt]
  \HilbertSpace_{p_{{}_0}}
  \ar[
    in=45+90,
    out=-45+90,
    looseness=4.7,
    shift left=2pt,
    shorten <=-2pt,
    "{ U_\ell }"
  ]
  \\[-20pt]
  &
  p_2
  \ar[
    dr,
    decorate,
    shorten <=-2pt,
    shorten >=-4pt,
    "{ \gamma_{{}_{23}} }"{yshift=2pt, sloped}
  ]
  &
  \\[-11pt]
  p_1
  \ar[
    rr, 
    decorate,
    shift right=1pt,
    shorten <=-2pt,
    shorten >=-3pt,
    "{
      \gamma_{{}_{13}}
    }"{yshift=2pt}
  ]
  \ar[
    ur,
    decorate,
    shorten <=-2pt,
    shorten >=-2pt,
    "{ 
      \gamma_{{}_{12}} 
    }"{yshift=2pt, sloped}
  ]
  &&
  p_{{}_3}
  &
  p_{{}_0}
  \ar[
      in=52+90,
      out=-55+90,
      looseness=5,
      shift left=4pt,
      decorate,
      shorten <=-2pt,
      "{ \ell }"
  ]
\end{tikzcd}
$
\hspace{-15pt}
}
\end{SCfigure}
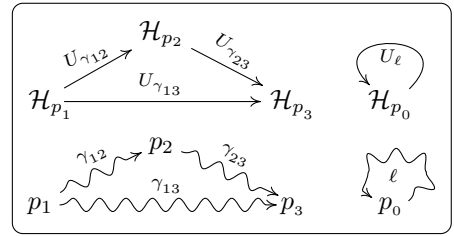

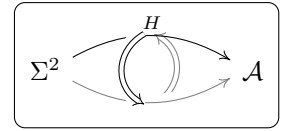
\begin{SCfigure}[3][htb]
\caption{\label{BlochMonodromySchematics}
For interacting electrons in crystalline materials, their external parameter is the crystal lattice structure encoded in the single-electron Bloch Hamiltonian $H$, hence the parameter monodromy of their quantum ground states (Fig. \ref{BlochMonodromySchematics}) is along loops in the space of Bloch Hamiltonians.
}

\adjustbox{
  rndfbox=4pt
}{
  \begin{tikzcd}[column sep=30pt]
    \Sigma^2
    \ar[
      rr, 
      bend left=30,
      phantom,
      "{}"{name=s, swap}
    ]
    \ar[
      rr, 
      bend right=30,
      phantom,
      shift left=3pt,
      "{}"{name=t}
    ]
    \ar[
      from=t, to=s,
      shorten=-4pt,
      Rightarrow,
      shift right=4pt,
      bend right=60,
      gray,
      scale=.5
    ]
    \ar[
      rr, 
      bend left=30,
      "{ H }",
    ]
    \ar[
      rr, 
      bend right=30,
      gray,
      shift left=3pt,
      crossing over
    ]
    \ar[
      from=s, to=t,
      shorten=-4pt,
      Rightarrow,
      shift right=4pt,
      bend right=60,
      crossing over
    ]
    &&
    \mathcal{A}
  \end{tikzcd}
}

\end{SCfigure}

For global analysis, the momentum domain $\Sigma^2$ must be the full Brillouin torus $\Sigma^2 = \widehat{T}^2$.
\begin{example}
\label[example]{HeisenbergGroupsAsPi1}
In the case of 2-band (fractional) Chern insulators \eqref{S2Classifying2BandChernInsulators} this yields \cite{KSS26-HigherAnyons}, as parameter monodromy in the topological phase with Chern class $C \in \mathbb{C}$, the \emph{integer Heisenberg group} at level 2 and Planck constant $\hbar = 2C$
\begin{equation}
  \pi_1\bracket({
    \mathrm{Maps}\bracket({
      \widehat{T}^2
      ,
      \mathbb{C}P^1
    }),
    C
  })
  \simeq
  \mathrm{Heis}_3\bracket({
    \mathbb{Z}, 2C
  })
  \mathrlap{\,.}
\end{equation}
Remarkably, this is exactly the group whose representations give fractional quantum Hall anyon states on the torus (\parencites[Thm. 1]{SS25-FQAH}[\S 3]{SS25-FQH}, review in \cite{SS25-ISQS29}).
\end{example}

\subsection{Local Bloch Hamiltonian Monodromy}
\label{OnLocalBlochHamiltonianMonodromy}

However, in the following we consider just the local situation, where (as suggested in \cite{BouhonEtAl2020,SS22-Ord,Bouhon2023}) $\Sigma^2$ is taken to be a \emph{punctured} submanifold of the Brillouin torus, where we think of the punctures as exhibiting (externally fixed) defect loci, such as a band nodes where the gapped band structure degenerates.

Specifically, we consider (in \cref{OnSimpleObjects}) $\Sigma^2$ to be the \emph{annulus}, hence the once-punctured version
\begin{equation}
  \label{TheOncePuncturedDisk}
  \Sigma^2
  :=
  D^2_{< 1} \setminus \{0\}
\end{equation}
of the open disk in the complex plane
$$
  D^2_{< r}
  :=
  \bracketmid\{{
    z \in \mathbb{C}
  }{
    \bracket\vert{z}\vert < r
  }\}
  \mathrlap{\,,}
$$
understood as modelling (only) the \emph{local} situation in the vicinity of a defect in momentum space. However, towards globalizing this picture, we crucially consider (in \cref{OnFusion}) also the situation where a pair of punctures approaches each other, reflected by the following cospan diagram (cf. \cref{TheCobordism}):
\begin{equation}
  \label{CospanOfPuncturedDisks}
  \begin{tikzcd}[
    column sep=-40pt
  ]
    &
    \bracket({
      D^2_{<1} 
        \setminus 
      \{-\sfrac{1}{3}, +\sfrac{1}{3}\}
    })
    \\
    \substack{
      \bracket({
        D^2_{< 1}
          \setminus \{-\sfrac{1}{3}\}
      })
      \\
      \mathllap{\sqcup\;}
      \bracket({
        D^2_{< 1}
          \setminus \{+\sfrac{1}{3}\}
      })
    }
    \ar[ur, hook]
    &&
    \substack{
      \bracket({
        D^2_{< 1}
          \setminus
        D^2_{\leq 2/3}
      })
      \\
      \mathllap{\sim \;}
      \bracket({
        D^2_{< 1}
        \setminus
        \{0\}
      })
      \,.
    }
    \ar[
      ul,
      hook'
    ]
  \end{tikzcd}
  \hspace{.8cm}
  \adjustbox{}{
\begin{tikzpicture}[
   baseline=
    {([yshift=-2pt]current bounding box.center)},
]
    \draw[thick, fill=lightgray] 
      (0,0) circle (1cm);

   \draw[
     fill=lightgray!30
   ]
    (0,0) circle (.7);

    \node at (1.05, .95) {$D^2_{<1}$};
    
    \draw[->, gray, thin] 
      (-1.3,0) -- (1.8,0) 
        node[
          pos=.9, 
          below,
          scale=.8
        ] 
          {$\mathrm{Re}(z)$};

    \draw[->, gray, thin] 
      (0,-1.3) -- (0,1.6) 
        node[
          pos=.9,
          left,
          scale=.8
        ] {$\mathrm{Im}(z)$};
    
    \draw[
      fill=white
    ]
      (-0.33,0) circle (2pt) 
        node[
          below, 
          yshift=-2pt,
          scale=.6
        ] 
          {$-\sfrac{1}{3}$};
    \draw[
      fill=white
    ]
      (0.33,0) circle (2pt) 
        node[
          below, 
          yshift=-2pt,
          scale=.6
        ] 
         {$+\sfrac{1}{3}$};

\end{tikzpicture}  
  }
\end{equation}

\section{
  \texorpdfstring{Local Parameter Monodromy via the Drinfeld Center $\mathcal{Z}(G)$}
  {Local Parameter Monodromy via the Drinfeld Center Z(G)}
}
\label{LocalParameterMonodromyViaCenter}

It is now a matter of homotopical topology (cf. \cref{OnHomotopy}) to work out the parameter monodromy from \cref{OnHomotopyOfSpacesOfBlochHamiltonians}, hence the potential topological order, in the local situation \eqref{TheOncePuncturedDisk}, namely to compute the groups
\begin{equation}
  \pi_1\bracket({
    \mathrm{Map}\bracket({
      \Sigma^2,
      \mathcal{A}
    }),
    p
  })
  \defneq
  \pi_1\bracket({
    \mathrm{Map}\bracket({
      D^2_{< 1} \setminus \{0\}
      ,
      \mathcal{A}
    }),
    p
  })
  \mathrlap{\,.}
\end{equation}
We will find (in \cref{OnSimpleObjects}) that these local parameter monodromies around a defect point correspond exactly to the simple objects in the Drinfeld center fusion category $\mathcal{Z}\bracket({\mathrm{Vec}_G})$, for $G := \pi_1(\mathcal{A})$, and (in \cref{OnFusion}) that, as a pair of defect approaches each other, the local monodromies fuse according to the fusion rules of this category.

\subsection{
\texorpdfstring
  {Simple objects of $\mathcal{Z}(G)$ as Local Parameter Monodromy}
  {Simple objects of Z(G) as Local Parameter Monodromy}
}
\label{OnSimpleObjects}

Given a choice of classifying space $\mathcal{A}$ (\cref{OnClassifyingSpacesOfBlochHamiltonians}), which we assume to be path-connected, $\pi_0(\mathcal{A}) \simeq \ast$, we will be concerned with its based and free \emph{loop spaces} \cref{BasedAndFreeLoopSpace}:
\begin{equation}
  \Omega
  \mathcal{A}
  :=
  \mathrm{Map}^\ast\bracket({
    S^1, 
    \mathcal{A}
  })
  \mathrlap{\,,}
  \quad 
  \mathcal{L}
  \mathcal{A}
  :=
  \mathrm{Map}\bracket({
    S^1, 
    \mathcal{A}
  })
  \mathrlap{\,.}
\end{equation}
The point is that, up to homotopy equivalence $\sim$ \cref{HomotopyEquivalenceRelation}, the free loop space is also the space of maps from the 1-punctured disk \eqref{TheOncePuncturedDisk}:
\begin{equation}
  \label{FreeLoopSpaceAsMapsFromAnnulus}
  \mathcal{L}\mathcal{A}
  \sim
  \mathrm{Map}\bracket({
    D^2\setminus \{0\},
    \mathcal{A}
  })
  \mathrlap{\,.}
\end{equation}
As such, by the discussion in \cref{OnHomotopyOfSpacesOfBlochHamiltonians}, we may think of:
\begin{enumerate}
\item 
$\mathcal{L}\mathcal{A}$ as the space of Bloch Hamiltonians on an open neighborhood of a puncture in the Brillouin torus;

\item
$[g] \in  \pi_0\bracket({\mathcal{L}\mathcal{A}})$ as a ``local topological phase'';
\footnote{
Here the adjective ``local'' refers to the fact \eqref{FreeLoopSpaceAsMapsFromAnnulus} that we are dealing with Bloch Hamiltonians only locally in the vicinity of a puncture.}

\item
$\pi_1\bracket({\mathcal{L}\mathcal{A}, [g]})$ as the ``local parameter monodromy'' in the local topological phase $[g]$.

\end{enumerate}

We want to describe this situation in terms of the \emph{fundamental group} of $\mathcal{A}$, which is:
\begin{equation}
  \label{FundamentalGroup}
  G
  :=
  \pi_1(\mathcal{A})
  \defneq
  \pi_0\bracket({\Omega\mathcal{A}})
  \,.
\end{equation}
First, note that: 
\begin{lemma}
\label[lemma]{ConjugacyClassesAsFundamentalGroupOfFreeLoopSpace}
The local topological phases form the set of conjugacy classes of elements of $G$: 
\begin{equation}
  \label{ConjugacyClassesAsComponentsOfFreeLoopSpace}
  \begin{aligned}
    \pi_0\bracket({\mathcal{L}\mathcal{A}})
    &
    \simeq
    \mathrm{Conj}(G)
    \\
    & :=
    \bracketmid\{{
      \bracketmid\{{
        g^k
        :=
        k g k^{-1}
      }{
        k \in G
      }\}
    }{
      g \in G
    }\}
    \mathrlap{\,.}
  \end{aligned}
\end{equation}
\end{lemma}
\begin{proof}
This is standard, but for the following discussion it is worth making the proof explicit:
A path $\inlinetikzcd{\widehat{\gamma} : [0,1] \ar[r] \& \mathcal{L}\mathcal{A} }$ in the free loop space moves the loop's base-point along an underlying path in $\mathcal{A}$, $\inlinetikzcd{ \gamma : [0,1] \ar[r] \& \mathcal{A} }$, whence $\widehat{\gamma}(1)$ is based homotopic to the result of conjugate-concatenating $\gamma$ with $\widehat{\gamma}(0)$, like this (cf. \cref{PathOfLoops}):
\begin{equation}
\label{EndpointLoopConjugateToStartingLoop}
\widehat{\gamma}(1) 
    \sim_{\mathrm{based}}
  \gamma 
    \star 
  \widehat{\gamma}(0) 
    \star 
  \overline{\gamma}
  \mathrlap{\,.}
\end{equation}
This way, first all elements of $\mathcal{L}\mathcal{A}$ are seen to be path-connected to based loops (recalling that we assume $\mathcal{A}$ to be connected), and then among these based loops precisely those are path-connected  which correspond to conjugate elements in $\pi_1(\mathcal{A})$.
\end{proof}

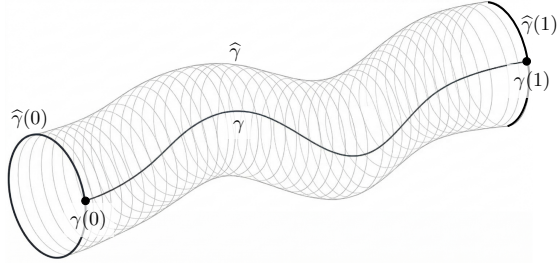
\begin{SCfigure}[.6][htb]
\caption{\label{PathOfLoops}
 A path $\widehat{\gamma}$ in a free loop space $\mathcal{L}\mathcal{A}$ and the path $\gamma$ traced out by the basepoint of the loops in $\mathcal{A}$. The loop $\widehat\gamma(1)$ is based-homotopic to the conjugate concatenation of $\widehat\gamma(0)$ by $\gamma$, cf. \cref{EndpointLoopConjugateToStartingLoop}.
}

\begin{tikzpicture}

  \node at (0,0) {
  \includegraphics
    [width=7cm]
    {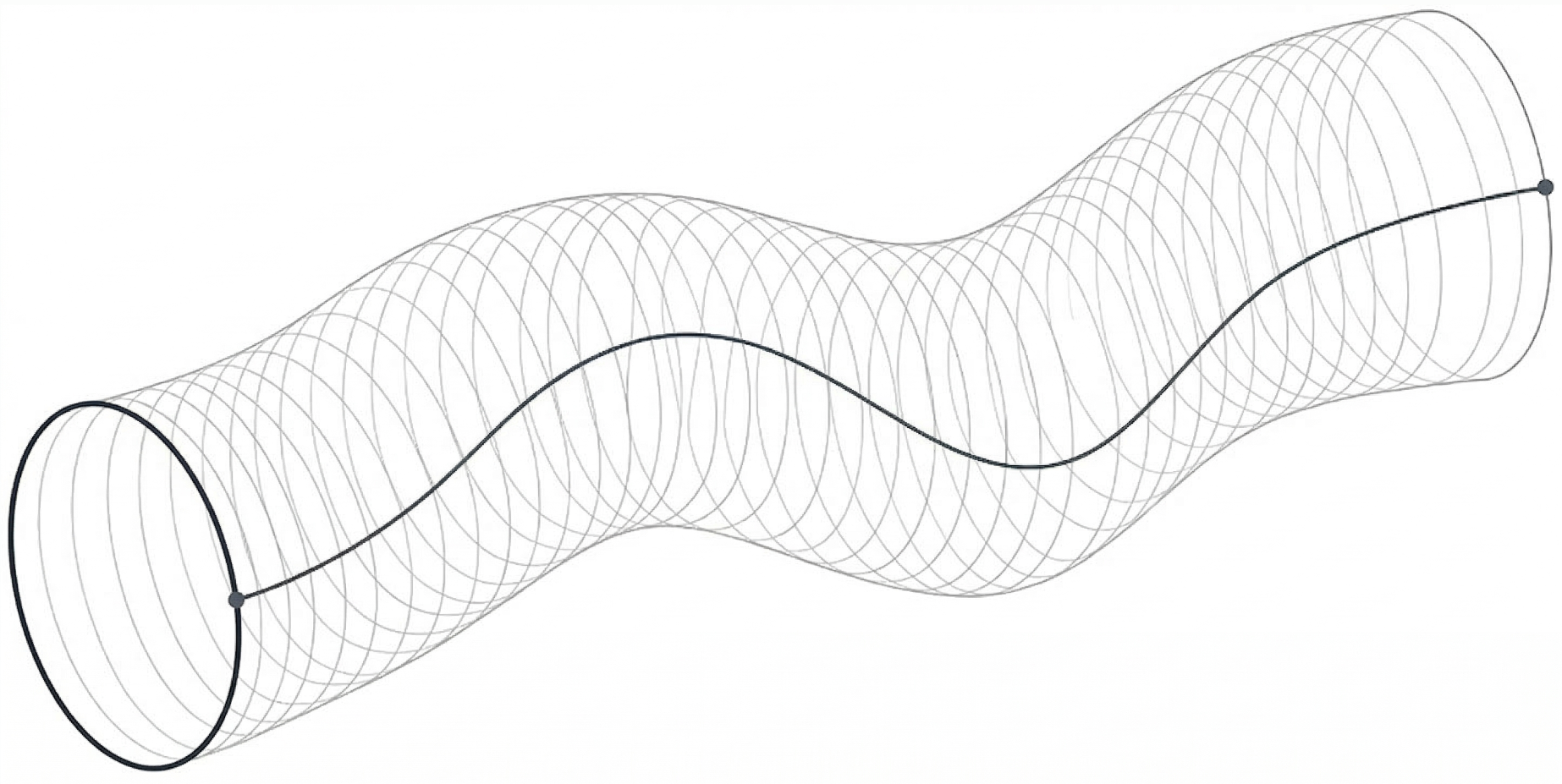}};

  \draw[
    rotate=9.8,
    shift={(3.095,.35)},
    line width=.75
  ]
    (-85:.452 and .84) arc 
    (-85:82:.452 and .84);

  \node[
    scale=.75
  ] at (-3.18,.15) {$\widehat\gamma(0)$};
  \node[
    scale=.75
  ] at (-.45,1.05) {$\widehat\gamma$};
  \node[
    scale=.75
  ] at (3.6,1.4) {$\widehat\gamma(1)$};

  \node[
    scale=.75, 
    fill=white,
    fill opacity=.7,
    text opacity=1
  ] at 
    (-.4,.05) {$\gamma$};
  \node[
    scale=.75, 
    fill=white,
    fill opacity=.7,
    text opacity=1
  ] at 
    (-2.4,-1.2) {$\gamma(0)$};
  \node[
    scale=.75, 
    fill=white,
    fill opacity=.7,
    text opacity=1
  ] at 
    (3.5,.65) {$\gamma(1)$};

  \draw[fill=black]
    (-2.44,-.94) circle (.05);

  \draw[fill=black]
    (3.425,.916) circle (.05);

\end{tikzpicture}
\end{SCfigure}

It remains to determine the parameter monodromy groups $\pi_1\bracket({\mathcal{L}\mathcal{A}})$ in each of these phases:

\begin{proposition}
  \label[proposition]{pi1LA}
  In every topological phase $[g] \in \mathrm{Conj}(G) \simeq \pi_0(\mathcal{L}\mathcal{A})$ \eqref{ConjugacyClassesAsComponentsOfFreeLoopSpace} we have:
  \begin{enumerate}
  \item 
  Generally, the image of the parameter monodromy $\pi_1\bracket({\mathcal{L}\mathcal{A}})$ in $G$ maps isomorphically onto the centralizer group \cref{CentralizersAsIsotropyGroups}
  of $g$:
  \begin{equation}
    \label{ImageOfTopologicalMonodromyInG}
    \pi_1(\mathrm{ev})
    \bracket({
      \pi_1\bracket({
        \mathcal{L}\mathcal{A}
      })
    })
    \simeq
    Z_G(g)
    \subset 
    G
    \mathrlap{\,.}
  \end{equation}
  \item 
  But when $\mathcal{A}$ has trivial $\pi_2$, then the parameter monodromy already coincides with the centralizer:
  \begin{equation}
    \label{ParameterMonodromyCoincidingWithCentralizer}
    \mathllap{
    \pi_2(\mathcal{A}) \simeq \ast
    \;\;\;\;\;
    \Rightarrow
    \;\;\;\;\;
    }
    \pi_1\bracket({
      \mathcal{L}\mathcal{A}
    })
    \simeq
    Z_G(g)
    \mathrlap{\,.}
  \end{equation}
  \end{enumerate} 
\end{proposition}
\begin{proof}
Consider any $\gamma \in \Omega \mathcal{A}$ which represents the local topological phase $[g]$: 
$$
  \begin{tikzcd}[row sep=-2pt, column sep=0pt]
    G
    \ar[rr, "{ \sim }"]
    &&
    \pi_0\bracket({\Omega\mathcal{A}})
    \ar[rr]
    &&
    \pi_0\bracket({
      \mathcal{L}\mathcal{A}
    })
    \ar[rr, "{ \sim }"]
    &&
    \mathrm{Conj}(G)
    \\
    g
    &\leftrightarrow&
    {[\gamma]}
    &\mapsto&
    {[\gamma]} 
      &\leftrightarrow& 
    {[g]}
    \mathrlap{\,.}
  \end{tikzcd}
$$
With that used as the base point, consider the homotopy long exact sequence (\cref{HomotopyLES})
induced by the map $\mathrm{ev}$ \cref{EvaluationSerreFibration} that evaluates a loop at its base point:
\begin{equation}
  \label{EvaluationFiberSequence}
  \begin{tikzcd}
    \Omega\mathcal{A}
    \ar[
      r, 
      "{\iota}",
      hook
    ]
    &
    \mathcal{L}\mathcal{A}
    \ar[
      r, 
      ->>,
      "{ \mathrm{ev} }"
    ]
    &
    \mathcal{A}
    \mathrlap{\,,}
  \end{tikzcd}
\end{equation}
which is of this form:
\begin{equation}
  \label{HomotopyExactSequence}
  \begin{tikzcd}
    \pi_1\bracket({
      \Omega\mathcal{A}
      ,
      \gamma
    })
    \ar[
      r
    ]
    &
    \pi_1\bracket({
      \mathcal{L}\mathcal{A}
      ,
      \gamma
    })
    \ar[
      r,
      "{ \pi_1(\mathrm{ev}) }"
    ]
    &
    \pi_1\bracket({
      \mathcal{A}
    })
    \ar[
      dll,
      snake left,
      "{
        \partial_0
      }"{description}
    ]
    \\
    \pi_0\bracket({\Omega \mathcal{A}})
    \ar[
      r
    ]
    &
    \pi_0\bracket({\mathcal{L}\mathcal{A}})
  \end{tikzcd}
  \;\;=\;\;
  \begin{tikzcd}
    \pi_1\bracket({
      \Omega\mathcal{A}
      ,
      \gamma
    })
    \ar[
      r
    ]
    &
    \pi_1\bracket({
      \mathcal{L}\mathcal{A}
      ,
      \gamma
    })
    \ar[
      r,
      "{ \pi_1(\mathrm{ev}) }"
    ]
    &
    G
    \ar[
      dll,
      snake left,
      "{
        \mathrm{Ad}_{(-)}(g)
      }"{description}
    ]
    \\
    G
    \ar[r, ->>]
    &
    \mathrm{Conj}(G)
    \mathrlap{\,.}
  \end{tikzcd}
\end{equation}
(Here the symbol ``$G$'' in the bottom right denotes the underlying set of the group, pointed by the element $g$, cf. \cref{TheHomotopyLESInLowDegree}.)

On the right of \eqref{HomotopyExactSequence} we are claiming that the connecting homomorphism $\partial_0$ \cref{ConnectingHomomorphism} acts by conjugation on $g \in G$. To see this, recall \cref{ConnectingHomomorphism} that $\partial_0$ is generally given on the class of a based loop $\ell \in P_0 \mathcal{A}$ by first lifting it through $\mathrm{ev}$ to a based path $\widehat \ell \in P_{[g]} \mathcal{L}\mathcal{A}$ and then evaluating that at its endpoint:
$
  \partial_0\bracket({[\ell]})
  =
  \bracket[{\widehat{\ell}(1)}]
$.
Here we may take $\widehat{\ell}$ to be given (cf. again \cref{PathOfLoops}) by
\begin{equation}
  \widehat{\ell}(t)(-)
  :=
    \ell(t-)
    \star
    \gamma
    \star
    \overline{\ell}(t-),
\end{equation}
which implies the intermediate claim \eqref{HomotopyExactSequence}. 

From this, the first claim \eqref{ImageOfTopologicalMonodromyInG} follows by exactness: The image of $\pi_1(\mathrm{ev})$ is now identified with the kernel of $\mathrm{Ad}_{(-)}(g)$, and that is $Z_G(g)$, by definition. 

Similarly for the second claim \eqref{ParameterMonodromyCoincidingWithCentralizer}: If $\pi_2(\mathcal{A}) = \pi_1(\Omega \mathcal{A})$ is trivial, then exactness gives that $\pi_1(\mathrm{ev})$ is injective and hence an isomorphism onto its image.
\end{proof}

In words, we may summarize \cref{ConjugacyClassesAsFundamentalGroupOfFreeLoopSpace,pi1LA} as follows, recalling that quantum states falling into irreducible representations constitute the \emph{superselection sectors} of a quantum system (cf. \parencites[Def. 2.1]{FrohlichEtAl1990}[p. 273]{Baker2013}), and assuming now that $\pi_2(\mathcal{A})$ is trivial (as in \cref{PTSymmetric3BandSystems}):

\begin{standout}
  The topological phases in the vicinity of a defect are labeled by conjugacy classes $[g] \in \mathrm{Conj}(G)$, and the superselection sectors of the topological orders in such a phase are labeled by irreps $\rho \in \mathrm{Irr}\bracket({Z_G(g)})$ of the centralizer.
\end{standout}

Strikingly, this exactly coincides with the characterization (cf. \cref{DrinfeldCenterAsInertiaGroupoidRepresentations}) of the simple objects in the Drinfeld center fusion category $\mathcal{Z}\bracket({\mathrm{Vec}_G})$. Hence we have equivalently:
\begin{standout}
  The superselection sectors of topological orders near defect points are labeled by the simple objects of the fusion category $\mathcal{Z}\bracket({\mathrm{Vec}_G})$.
\end{standout}

Our next goal is to show that also the fusion rules match.

\subsection{
\texorpdfstring
  {Fusion in $\mathcal{Z}(G)$ as Fusion of Local Parameter Monodromy}
  {Fusion in Z(G) as Fusion of Local Parameter Monodromy}
}
\label{OnFusion}

To describe the situation of a pair of nearby defect points, consider the space of maps from the 2-punctured disk \cref{CospanOfPuncturedDisks}, to be denoted:
\begin{equation}
  \label{MapsOutOf2PuncturedDisk}
  \mathcal{L}^{(2)}\mathcal{\mathcal{A}}
  :=
  \mathrm{Map}\bracket({
    D^2
    \setminus
    \{
      +\sfrac{1}{3}
      ,
      -\sfrac{1}{3}
    \}
    ,
    \mathcal{A}
  })
  \mathrlap{\,.}
\end{equation}

First, we consider the homotopy of this space. In an immediate variation of \cref{ConjugacyClassesAsFundamentalGroupOfFreeLoopSpace}, we have:
\begin{lemma}
The connected components of $\mathcal{L}^{(2)}\mathcal{A}$ \cref{MapsOutOf2PuncturedDisk} correspond to joint conjugacy classes of pairs of group elements:
\begin{equation}
  \label{DoubleConjugacyClassesAsFundamentalGroupOfDoubleFreeLoopSpace}
  \begin{aligned}
  \pi_0\bracket({
    \mathcal{L}^{(2)}\mathcal{\mathcal{A}}
  })
  & \simeq
  \mathrm{Conj}^{(2)}\bracket({G})
  \\
  & :=
  \bracketmid\{{
    \bracketmid\{{
      \bracket({
        g_1^k, g_2^k
      })
      :=
      \bracket({
        k g_1 k^{-1}, k g_2 k^{-1}
      })
    }{
      k \in G
    }\}
  }{
    g_1, g_2 \in G
  }\}
  \mathrlap{\,.}
  \end{aligned}
\end{equation}
\end{lemma}

Similarly, the immediate but crucial variation of \cref{pi1LA} is:
\begin{proposition}
  \label[proposition]{pi1L2A}
  In every phase $[g_1,g_2] \in \mathrm{Conj}^{(2)}(G) \simeq \pi_0\bracket({\mathcal{L}^{(2)}\mathcal{A}})$ \cref{DoubleConjugacyClassesAsFundamentalGroupOfDoubleFreeLoopSpace} we have: 
  \begin{enumerate}
  \item
  Generally, the image of the parameter monodromy $\pi_1\bracket({\mathcal{L}\mathcal{A}})$ in $G$ is isomorphic into the joint centralizer of $(g_1,g_2)$:
  \begin{equation}
    \pi_1(\mathrm{ev})
    \bracket({
      \pi_1\bracket({
        \mathcal{L}^{(2)}
        \mathcal{A}
      })
    })
    \simeq
    Z_G(g_1)
    \cap
    Z_G(g_2)
    \subset G
    \mathrlap{\,.}
  \end{equation}

  \item  
  But when $\mathcal{A}$ has trivial $\pi_2$, then the parameter monodromy already coincides with the joint centralizer:
  \begin{equation}
    \mathllap{
      \pi_2(\mathcal{A})
      \simeq
      \ast
      \;\;\;\;
      \Rightarrow
      \;\;\;\;
    }
    \pi_1\bracket({
      \mathcal{L}^{(2)}
      \mathcal{A}
    })
    \simeq
    Z_G(g_1)
    \cap
    Z_G(g_2)
    \mathrlap{\,.}
  \end{equation}
  \end{enumerate} 
\end{proposition}
\begin{proof}
Consider a pair of loops $(\gamma_1, \gamma_2)$ which represent $[g_1,g_2] \in G^2/_{\!\sim}$:
$$
  \begin{tikzcd}[row sep=-2pt, 
    column sep=0pt
  ]
    G^2
    \ar[
      rr, 
      "{ \sim }"
    ]
    &&
    \pi_0\bracket({
      (\Omega\mathcal{A})^2
    })
    \ar[
      rr,
      "{  }"
    ]
    &&
    \pi_0\bracket({
      \mathcal{L}^{(2)}
      \mathcal{A}
    })
    \ar[
      rr,
      "{ \sim }"
    ]
    &&
    G^2/_{\!\sim}
    \\
    (g_1,g_2)
    &\leftrightarrow&
    {[\gamma_1,\gamma_2]}
    &\mapsto&
    {[\gamma_1,\gamma_2]}
    &\leftrightarrow&
    {[g_1,g_2]}\,.
  \end{tikzcd}
$$
With that used as the base point, consider the homotopy long exact sequence (cf. \cref{HomotopyLES}) induced by the evaluation map \cref{EvaluationSerreFibration}
\begin{equation}
  \begin{tikzcd}
    (\Omega\mathcal{A})^2
    \ar[
      r,
      hook,
      "{ \iota }"
    ]
    &
    \mathcal{L}^{(2)}
    \mathcal{A}
    \ar[
      r,
      ->>,
      "{ \mathrm{ev} }"
    ]
    &
    \mathcal{A}
    \mathrlap{\,,}
  \end{tikzcd}
\end{equation}
of the form:
\begin{equation}
  \begin{tikzcd}
    &[-25pt]
    \pi_1\bracket({
      \mathcal{L}^{(2)}
      \mathcal{A},
      (\gamma_1,\gamma_2)
    })
    \ar[
      r,
      "{
        \pi_1(\mathrm{ev})
      }"
    ]
    &
    \pi_1(\mathcal{A})
    \ar[
      dll,
      snake left,
      "{ 
        \partial 
      }"{description}
    ]
    \\
    \pi_0\bracket({
      (\Omega\mathcal{A})^2
    })
    \ar[
      r,
      ->>
    ]
    &
    \pi_0\bracket({
      \mathcal{L}^{(2)}
      \mathcal{A}
    })    
  \end{tikzcd}
  \;\;
  =
  \;\;
  \begin{tikzcd}
    &[-35pt]
    \pi_1\bracket({
      \mathcal{L}^{(2)}
      \mathcal{A},
      (\gamma_1,\gamma_2)
    })
    \ar[
      r,
      "{
        \pi_1(\mathrm{ev})
      }"
    ]
    &
    G
    \ar[
      dll,
      snake left,
      "{ 
        \mathrm{Ad}_{(-)}(
          g_1,g_2
        )
      }"{description}
    ]
    \\
    G^2
    \ar[
      r,
      ->>
    ]
    &
    G^2/_{\!\sim}
    \mathrlap{\,.}
  \end{tikzcd}
\end{equation}
Here ``$G^2$'' on the right denotes the underlying set of pairs of groups elements, cf. \cref{TheHomotopyLESInLowDegree}, pointed by $(g_1,g_2)$, and $\mathrm{Ad}_{k}(g_1,g_2) := \bracket({ k g_1 k^{-1}, k g_2 k^{-1} })$.

This identification follows by the same logic as in the proof of Prop. \ref{pi1LA}, the lifted path now being
$$
  \widehat{\ell}_t
  :=
  \bracket({
    \ell(t-)
    \star
    \gamma_1
    \star
    \overline{\ell}(t-)
    ,\,
    \ell(t-)
    \star
    \gamma_2
    \star
    \overline{\ell}(t-)
  })
  \mathrlap{\,.}
$$

The claim follows by exactness, as before.
\end{proof}

Now consider a topological phase $[g_1,g_2] \in \mathrm{Conj}^{(2)}(G) \simeq \pi_0\bracket({\mathcal{L}^{(2)}\mathcal{A}})$ \cref{DoubleConjugacyClassesAsFundamentalGroupOfDoubleFreeLoopSpace} locally around the pair of defects. Then, for every choice of representatives
$$
  (\gamma_1, \gamma_2)
  \in \Omega\mathcal{A}\,,
$$
we obtain a span of groups, by homming the cospan \eqref{CospanOfPuncturedDisks} into $\mathcal{A}$ and passing to $\pi_1$:
\begin{equation}
  \label{SpanOfFundamentalGroups}
  \begin{tikzcd}[
    column sep=-35pt
  ]
    &
    \pi_1\bracket({
      \mathcal{L}^{(2)}
      \mathcal{A},
      (\gamma_1,\gamma_2)
    })
    \ar[
      dl
    ]
    \ar[
      dr,
    ]
    \\
    \substack{
      \phantom{\times}\,
      \pi_1\scaledbracket({
        \mathcal{L}\mathcal{A},
        \gamma_2
      })
      \\
      \times
      \pi_1\scaledbracket({
        \mathcal{L}\mathcal{A},
        \gamma_1
      })
    }
    &&
    \pi_1\bracket({
      \mathcal{L}\mathcal{A},
      \gamma_2\star\gamma_1
    })
  \end{tikzcd}
  :=
  \pi_1\,
  \mathrm{Map}\left(
  \begin{tikzcd}[
    column sep=-40pt
  ]
    &
    \bracket({
      D^2_{<1} 
        \setminus 
      \{-\sfrac{1}{3}, +\sfrac{1}{3}\}
    })
    \\
    \substack{
      \phantom{\sqcup}\,
      \bracket({
        D^2_{< 1}
          \setminus \{-\sfrac{1}{3}\}
      })
      \\
      \sqcup\,
      \bracket({
        D^2_{< 1}
          \setminus \{+\sfrac{1}{3}\}
      })
    }
    \ar[ur, hook]
    &&
    \substack{
      \bracket({
        D^2_{< 1}
          \setminus
        D^2_{\leq 2/3}
      })
      \\
      \mathllap{\sim \;}
      \bracket({
        D^2_{< 1}
        \setminus
        \{0\}
      })
    }
    \ar[
      ul,
      hook'
    ]
  \end{tikzcd}
  \;,\;
  \mathcal{A}
  \right) 
  \mathrlap{.}
\end{equation}
This encodes how the parameter monodromy in the vicinity of each defect separately relates to that of their joint vicinity. Concretely:
\begin{lemma}
The span of groups \cref{SpanOfFundamentalGroups} is:
\begin{equation}
  \label{EvaluatingTheSpanOfGroups}
  \begin{tikzcd}[
    column sep=-40pt
  ]
    &
    \pi_1\bracket({
      \mathcal{L}^{(2)}
      \mathcal{A},
      (\gamma_1,\gamma_2)
    })
    \ar[
      dl
    ]
    \ar[
      dr,
    ]
    \\
    \substack{
      \phantom{\times}\,
      \pi_1\scaledbracket({
        \mathcal{L}\mathcal{A},
        \gamma_2  
      })
      \\
      \times
      \pi_1\scaledbracket({
        \mathcal{L}\mathcal{A},
        \gamma_1
      })
    }
    &&
    \pi_1\bracket({
      \mathcal{L}\mathcal{A},
      \gamma_2\star\gamma_1
    })
  \end{tikzcd}
  \;\;=\;\;
  \begin{tikzcd}[
   column sep=-20pt
  ]
    &[-20pt]
    Z_G(g_2)
    \cap
    Z_G(g_1)
    \ar[
      dl,
      hook',
      "{ i }"{swap}
    ]
    \ar[
      dr,
      hook,
      "{ o }"
    ]
    \\
    Z_G(g_2)
    \times
    Z_G(g_1)
    &&
    Z_G( g_2 g_1 )
    \mathrlap{\,.}
  \end{tikzcd}
\end{equation}
\end{lemma}
\begin{proof}
  This follows by  \cref{pi1LA,pi1L2A}.
\end{proof}

We still need to say how the quantum states in the vicinity of the separate defects relate to the quantum states in their joint vicinity, hence how quantum states ``fuse'' as the defects approach each other.  The natural mathematical reflection of this process is the ``pull-tensor-push'' or ``integral transform'' 
(cf. \parencites[Def. 7.6]{Sc14-LinTypes}[(148)]{SS25-Complete})
of quantum states through the span \cref{EvaluatingTheSpanOfGroups}, where 
\begin{enumerate}
  \item
  the ``pull''-operation $i^\ast$ 
  moves the quantum states around the pair of separate defects to the joint situation;
  \item 
  the tensor product combines them to states of the composite system;
  \item
  the ``push''-operation $o_!$ sums up the resulting contributions to quantum states around the single defect obtained by merging the separate defects:
\end{enumerate}
\begin{equation}
  \label{TheIntegralTransform}
  \begin{tikzcd}[
    column sep=5pt,
    row sep=-1pt
  ]
    \mathrm{Irr}\bracket({
      Z_G(g_2)
    })
    \times
    \mathrm{Irr}\bracket({
      Z_G(g_1)
    })
    \ar[
      rr,
      shorten=-2pt,
      "{ i^\ast }"
    ]
    &&
    \mathrm{Rep}\bracket({
      Z_G(g_2,g_1)
    })
    \ar[
      rr,
      shorten=-2pt,
      "{ o_! }"
    ]
    &&
    \mathrm{Rep}\bracket({
      Z_G(g_2 g_1)
    })
    \\
    (
      \rho_2,\rho_1
    )
    &\mapsto&
      i^\ast\bracket({\rho_2})
      \otimes
      i^\ast\bracket({\rho_1})
    &\mapsto&
     o_! \bracket({
      i^\ast\bracket({\rho_2})
      \otimes
      i^\ast\bracket({\rho_1})
    })
    \mathrlap{\,,}
  \end{tikzcd}
\end{equation}
where $i^\ast(-)$ denotes restriction of representations and $o_!(-)$ denotes (left) induction of representations (cf. \cref{OnInducedRepresentations}), and where we are now abbreviating
joint centralizers as
$
  Z_G(g_2,g_1)
  :=
  Z_G(g_2) \cap Z_G(g_1)
  \mathrlap{\,.}
$

This is the natural TQFT operation on parameter monodromy which is  associated with the trinion cobordism (two incoming punctures, one outgoing) illustrated on the right of \cref{CospanOfPuncturedDisks}, cf. \cref{TheCobordism}.

\begin{SCfigure}[2.5][htb]
\caption{\label{TheCobordism}
The fusion process of a pair of defects into a single defect, modeled as punctures in the Brillouin torus as in \cref{CospanOfPuncturedDisks}, traces out the cobordism known as the \emph{trinion} or \emph{pair-of-pants}. The quantum states are thereby transformed by the integral transform \cref{TheIntegralTransform} through the corresponding span of parameter monodromy groups \cref{SpanOfFundamentalGroups}.
}
\includegraphics
  [width=4cm]{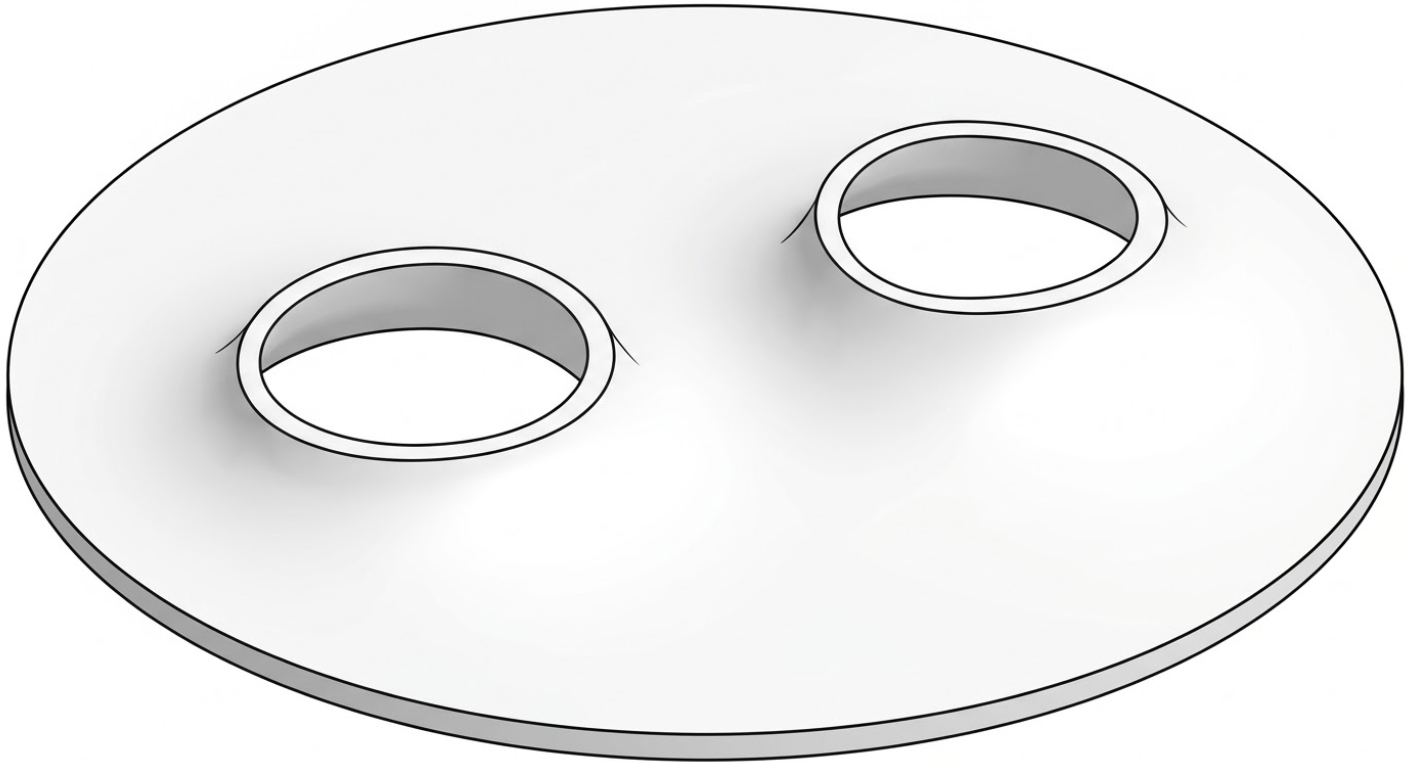}
\end{SCfigure}

But by Frobenius reciprocity, $o_! \dashv o^\ast$, this means that the multiplicity of an irrep $\rho \in \mathrm{Irr}\bracket({Z_G(g_2 g_1)})$ in the result of this pull-push is
\begin{equation}
  \mathrm{dim}
  \,
  \mathrm{Hom}
  \bracket({
    o_! i^\ast
    \bracket({
      \rho_2 \boxtimes \rho_1
    })
    ,\,
    \rho
  })
  =
  \mathrm{dim}
  \,
  \mathrm{Hom}
  \bracket({
    i^\ast
    \bracket({
      \rho_2 \boxtimes \rho_1
    })
    ,\,
    o^\ast\rho
  })
  \mathrlap{\,.}
\end{equation}
On the right this is the expression for the contribution, over a fixed decomposition $g = g_2 g_1$, to the fusion coefficient in the Drinfeld center of $G$-graded vector spaces (\cref{FusionCoefficientsInDrinfeldCenterOfGVectorSpaces}). 

So far, this describes the fusion of quantum states in a pair of fixed topological phases $[\gamma_1, \gamma_2]$. In order to account also for the fusion of these topological phases, we need to sum the above construction over all ways that the pair of topological phases $\bracket({[g_1], [g_2]})$ may fuse to the topological phase $[g_1 g_2]$. This count is established by the following (here we are abbreviating $g^k := k g k^{-1}$):
\begin{lemma}
  \label[lemma]{TheFusionMultiplicity}
  For $g_1, g_2 \in G$,
  the fiber over a conjugacy class $[g] \in \mathrm{Conj}(G)$ of the group multiplication map restricted to their $G$-orbits, $\inlinetikzcd{(-)\cdot(-) : [g_1]\times[g_2] \ar[r] \& G }$, and quotiented by the (diagonal) adjoint action, 
  is in bijection with the set
  \begin{equation}
    \label{DoubleCosetInProductOrbitDecomposition}
    \FusionMultiplicity_{g_1,g_2}^g
    :=
    \bracketmid\{{
      [k_1,k_2] \in
      Z_G(g) 
        \big\backslash  
      G^2 
        \big/ 
      \bracket({Z_G(g_1) \times Z_G(g_2)})
    }{
      g_1^{k_1} g_2^{k_2} = g
    }\}
    \mathrlap{\,.}
  \end{equation}
  in that:
  \begin{equation}
    \begin{tikzcd}
      R_{g_1,g_2}^g
      \ar[
        dr,
        phantom,
        "{ \lrcorner }"{pos=.1}
      ]
      \ar[
        r,
        hook
      ]
      \ar[
        d
      ]
      &
      \bracket({
        [g_1] \times [g_2]
      })\sslash_{\!\mathrm{Ad}} G
      \ar[
        d,
        "{
          (-)\cdot(-)
        }"
      ]
      \\
      {[g]}
      \ar[r, hook]
      &
      G \sslash_{\!\mathrm{Ad}} G
      \mathrlap{\,.}
    \end{tikzcd}
  \end{equation}
\end{lemma}
\begin{proof}
Unwinding the definitions, one finds that the following maps between $\FusionMultiplicity_{g_1,g_2}^g$ and the fiber set are well-defined:
\begin{equation}
  \begin{tikzcd}[
    row sep=0pt, 
    column sep=15pt
  ]
    \FusionMultiplicity_{g_1,g_2}^g
    \ar[
      rr,
      "{
        \phi
      }"
    ]
    &&
    \bracketmid\{{
      [q_1,q_2]
      \in
      \bracket({
        [g_1] \times [g_2]
      })\sslash_{\!\mathrm{Ad}} G
    }{
      [q_1 q_2] = [g] 
    }\}
    \\
    {[k_1,k_2]}
    \ar[
      rr,
      |->,
      shorten=6pt
    ]
    &&
    \bracketmid\{{
      ({
        g_1^{k_1}, g_2^{k_2}
      })^k
      :=
      \bracket({
        k k_1 g_1 k_1^{-1} k^{-1}
        ,
        k k_2 g_2 k_2^{-1} k^{-1}        
      })
    }{
      k \in G
    }\}
  \end{tikzcd}
\end{equation}
and
\begin{equation}
  \begin{tikzcd}[
    row sep=0pt, 
    column sep=15pt
  ]
    \bracketmid\{{
      [q_1,q_2]
      \in
      \bracket({
        [g_1] \times [g_2]
      })\sslash_{\!\mathrm{Ad}} G
    }{
      [q_1 q_2] = [g] 
    }\}
    \ar[
      rr,
      "{ \phi' }"
    ]
    &&
    \FusionMultiplicity_{g_1,g_2}^g
    \\
    {[q_1, q_2]}
    \ar[
      rr,
      |->,
      shorten=10pt
    ]
    &&
    \bracket[{
     \ln_{g_1}\bracket({
        q_1^{\ln_{q_1 q_2}(g)}
      }),
      \ln_{g_2}\bracket({
        q_2^{\ln_{q_1 q_2}(g)}
      })
    }]
    \mathrlap{\,,}
  \end{tikzcd}
\end{equation}
where $\ln_g(g') \in G$ denotes any choice of solution to the equation:
\begin{equation}
  g^{\ln_g(g')} 
  =
  g'
  \mathrlap{\,.}
\end{equation}
But these functions are inverse to each other, since
$$
  \begin{aligned}
  \phi'\bracket({
    \phi\bracket({
      [k_1, k_2]
    })
  })
  &
  \defneq
  \phi'\bracket({
    \bracket[{
      g_1^{k_1}, g_2^{k_2}
    }]
  })
  \\
  &
  \defneq
  \bracket[{
    \ln_{g_1}\bracket({g_1^{k_1}}),
    \ln_{g_2}\bracket({g_2^{k_2}})
  }]
  \\
  &
  =
  [k_1, k_2]
  \mathrlap{\,,}
  \end{aligned}
  \;\;\;\;
  \begin{aligned}
    \phi\bracket({
      \phi'\bracket({
        [q_1,q_2]
      })
    })
    & \defneq
    \phi\bracket({
      \bracket[{
       \ln_{g_1}\bracket({
          q_1^{\ln_{q_1 q_2}(g)}
        }),
        \ln_{g_2}\bracket({
          q_2^{\ln_{q_1 q_2}(g)}
        })
      }]
    })
    \\
    & \defneq
    \bracket[{
      q_1^{\ln_{q_1 q_2}(g)},
      q_2^{\ln_{q_1 q_2}(g)}
    }]
    \\
    & = 
    [{q_1, q_2}]
    \mathrlap{\,,}
  \end{aligned}
$$
which establishes the claimed bijection.
\end{proof}

In summary, this says that the fusion of a pair of local topological orders $\bracket({[g_1], \rho_1}), \bracket({[g_2],\rho_2})$ contains a superselection sector topological order $\bracket({[g], \rho})$ with multiplicity
\begin{equation}
  N_{
    ([g_1],\rho_1),
    ([g_2], \rho_2)
  }
  ^{
    ([g],\rho)
  }
  =
  \sum_{
    [k_1,k_2]
      \in
    \FusionMultiplicity_{g_1, g_2}^g
  }
  \!\! \mathrm{dim}
  \,
  \mathrm{Hom}\bracket({
    i^\ast \rho_1
    \otimes
    i^\ast \rho_2
    ,
    o^\ast \rho
  })
  \mathrlap{\,.}
\end{equation}

Remarkably, these are exactly the fusion coefficients in the Drinfeld center fusion category $\mathcal{Z}\bracket({\mathrm{Vec}_G})$, as we prove in \cref{FusionCoefficientsInDrinfeldCenterOfGVectorSpaces}.  In summary, expanding on the conclusion of \cref{OnSimpleObjects}:

\begin{standout}
  With the topological orders in the vicinity of defect points labeled by simple objects of $\mathcal{Z}\bracket({\mathrm{Vec}_G})$, the fusion of topological orders when defect points merge is described by the fusion rules of $\mathcal{Z}\bracket({\mathrm{Vec}_G})$. 
\end{standout}

\section{Conclusion}

After recalling (in \cref{OnHomotopyOfSpacesOfBlochHamiltonians}) that (anyonic) topological order is characterized by representations on gapped quantum ground states of the external parameter monodromy, we computed this monodromy for crystalline systems in the vicinity of a defect point in the Brillouin zone (set up in \cref{OnLocalBlochHamiltonianMonodromy}), as it may be relevant notably for recently established fractional Chern insulators. 

Under the assumption that the classifying space $\mathcal{A}$ of Bloch Hamiltonians (recalled in \cref{OnClassifyingSpacesOfBlochHamiltonians}) has trivial $\pi_2$ (as is the case in systems of interest, cf. \cref{PTSymmetric3BandSystems}), we found (in \cref{OnSimpleObjects}) that the superselection sectors (where the parameter monodromy representation is irreducible) correspond precisely to the simple objects of the Drinfeld center fusion category $\mathcal{Z}\bracket({\mathrm{Vec}_G})$ (recalled in \cref{OnTheDrinfeldCenter}), and (in \cref{OnFusion}) that under merging of defects these topological orders fuse according to the fusion rules of $\mathcal{Z}\bracket({\mathrm{Vec}_G})$ (which we establish in \S\ref{OnTheFusionProduct}).

This seems to be noteworthy, since (as recalled in \cref{Introduction}) fusion categories in general and the Drinfeld double model in particular are well-known from theoretical lattice models of anyonic topological order, but have not previously been explicitly related to the fractional crystalline topological phases that are currently the most prominent, if not the only, candidate platform for real-world topologically ordered quantum materials. 

In particular, contrary to common lattice models, the topological order analyzed here is localized not in position space but in the crystal's momentum space (around a defect point there, such as a band node), in line with the recent arguments of \parencites{SS25-FQAH}{SS25-Crys}.

Notably, by extension our analysis suggests/predicts that, besides the local quantum state monodromy in the vicinity of each defect point in momentum space, the adiabatic movement of the defect points themselves around each other should induce on the quantum states the braiding operation \cref{BraidingInDrinfeldCenter} of $\mathcal{Z}\bracket({\mathrm{Vec}_G})$ --- which is nonabelian when $G$ is nonabelian (as in \cref{PTSymmetric3BandSystems}), as required for universal topological quantum hardware. A similar result in FQH systems, of nonabelian braiding of defects, was recently argued in \cite{SS26-Islands}.

In this way, the theoretical analysis presented here may help guide ongoing experimental searches for topological order in quantum materials, and thus eventually for much anticipated future hardware for topologically stabilized quantum computers.

\appendix

\section{Background}

For reference in the main text, here we briefly set up and cite relevant background material.  Most is classical or well-known, but the specific account of Drinfeld fusion in \S\ref{OnTheFusionProduct} appears to be new, cf. \cref{LiteratureOnDrinfeldFusionProduct}.

\subsection{Homotopy}
\label{OnHomotopy}

We assume background in basic homotopical topology (cf. \cite{Hatcher2002,Fomenko2016,Arkowitz2011}); this here is just to briefly recall the phenomenon of \emph{homotopy long exact sequences} (\cref{HomotopyLES}).

We work in the category $\mathrm{Top}$ of compactly generated topological spaces. 

For a pair  $X, Y \in \mathrm{Top}$, the set of continuous maps between them becomes a topological space, $\mathrm{Map}\bracket({X,Y})$, via the \emph{compact open topology}. If both spaces are equipped with base points, $x_0 \in X$ and $y_0 \in Y$, then we write $\mathrm{Map}^\ast\bracket({X,Y})$ for the subspace of maps that respect these base points:
\begin{equation}
  \label{MappingSpace}
  \mathrm{Map}^\ast\bracket({X,Y})
  \subset
  \mathrm{Map}\bracket({X,Y})
  \mathrlap{\,.}
\end{equation}

For example, for
\begin{equation}
  I^n
  :=
  [0,1]^n \subset \mathbb{R}^n
\end{equation}
the closed $n$-cube, an \emph{$n$-path} in $Y$ is a continuous map $\inlinetikzcd{I^n \ar[r]\& Y}$ and this is \emph{based} if it takes the boundary $\partial\bracket({I^n})$ to the basepoint of $Y$. 
Endpoint evaluation of unbased 1-paths is an equivalence relation
\begin{equation}
  \begin{tikzcd}
    \mathrm{Map}\bracket({
      I,Y
    })
    \ar[
      rr,
      shift left=3pt,
      "{ (-)(0) }"
    ]
    \ar[
      rr,
      shift right=3pt,
      "{ (-)(1) }"{swap}
    ]
    &&
    Y
    \mathrlap{\,.}
  \end{tikzcd}
\end{equation}
The corresponding equivalence classes form the set of \emph{path connected components}
\begin{equation}
  \label{PathConnectedComponents}
  \pi_0(Y)
  \in
  \mathrm{Set}
  \mathrlap{\,.}
\end{equation}

Notably the \emph{homotopy class} $[f]$ of a map $\inlinetikzcd{X \ar[r, "{f}"] \& Y}$ is its connected component in the mapping space \cref{MappingSpace}, 
\begin{equation}
  [f]
  \in
  \pi_0\bracket({
    \mathrm{Map}({
      X,Y
    })
  })
  \mathrlap{\,,}
\end{equation}
and a map is a \emph{homotopy equivalence} if that class has an inverse class:
\begin{equation}
  \label{HomotopyEquivalence}
  \text{$f$ is homotopy equivalence}
  \;\;\;\;
  \Leftrightarrow
  \;\;\;\;
  \displaystyle{
  \exists \;
    \overline{f} 
      \in 
    \mathrm{Map}\scaledbracket({Y,X})}
    \; \mathrm{s.t.} \;
  \left\{
  \substack{
    g \circ f
    \,=\,
    \bracket[{
      \mathrm{id}_X
    }]
    \\  
    f \circ g
    \,=\,
    \bracket[{
      \mathrm{id}_Y
    }]
    \mathrlap{\,.}
  }
  \right.
\end{equation}
Existence of homotopy equivalences \cref{HomotopyEquivalence} between spaces is an equivalence relation, to be denoted
\begin{equation}
  \label{HomotopyEquivalenceRelation}
  X \sim Y
  \;\;\;\;
  \Leftrightarrow
  \;\;\;\; 
  \text{$\exists$ homotopy equivalence $\inlinetikzcd{X \ar[r] \& Y}$}
  \mathrlap{\,.}
\end{equation}

Similarly, the mapping spaces \cref{MappingSpace} of the circle $X \defneq S^1$ (equipped with any base point) gives the \emph{based loop space} and the \emph{free loop space} of $Y$, respectively:
\begin{equation}
  \label{BasedAndFreeLoopSpace}
  \Omega Y 
    := 
  \mathrm{Map}^\ast\bracket({S^1,Y})
  \,,
  \;\;\;
  \mathcal{L}Y
    :=
  \mathrm{Map}\bracket({S^1, Y})
  \mathrlap{\,.}
\end{equation}
Generally for  $X \defneq S^n$ the usual $n$-sphere (with any choice of base point), we have the \emph{$n$-fold based loop space}
\begin{equation}
  \label{BasedHigherLoopSpaces}
  \Omega^n X
    \simeq
  \mathrm{Map}^\ast\bracket({
    S^n,  X
  })
  \,.
\end{equation}

Under the evident homeomorphism
\begin{equation}
  S^n 
  \simeq
  I^n\big/\partial\bracket({I^n})
\end{equation}
an element of $\mathrm{Map}^{(\ast)}\bracket({S^n, Y})$ is equivalently a (based) $n$-path which is constant on $\partial\bracket({I^n})$. 

Now, with $x$ denoting the first and $\vec y$ denoting the remaining canonical coordinate functions on $I^n$, the \emph{concatenation} of a $k$-tuple of based $n$-paths $\big({\inlinetikzcd{\gamma_i : I^n \ar[r] \& Y }}\big)_{i = 0}^{k-1}$ is the based path
\begin{equation}
  \label{ConcatenationOfNPaths}
  \gamma_{(k-1)}
  \star
  \cdots
  \star
  \gamma_0
  \;:\;
  (x,\vec y\,)
  \longmapsto
  \begin{cases}
    \gamma_0\bracket({
      k x,
      \vec y\,
    })
    &
    \text{ 
      for $0 \leq x \leq \tfrac{1}{k}$ 
    }
    \\
    \gamma_1\bracket({
      k \bracket({x - \tfrac{1}{k}}),
      \vec y\,
    })
    &
    \text{ 
      for $\tfrac{1}{k} \leq x \leq \tfrac{2}{k}$ 
    }
    \\
    \gamma_i\bracket({
      k \bracket({x - \tfrac{i}{k}}),
      \vec y\,
    })
    &
    \text{ 
      for $\tfrac{i}{k} \leq x \leq \tfrac{i+1}{k}$,
    }
  \end{cases}
\end{equation}
and the \emph{reversal} of a path $\gamma$ is the path
\begin{equation}
  \label{PathReversal}
  \overline{\gamma}
  :
  \bracket({
    x, 
    \vec y\,
  })
  \longmapsto 
  \bracket({
    (1-x)
    ,
    \vec y\,
  })
  \mathrlap{\,.}
\end{equation}

Since the concatenation operation \cref{ConcatenationOfNPaths} is associative \emph{up to homotopy} and the reversal operation  \cref{PathReversal} is an inversion up to homotopy, the connected components \cref{PathConnectedComponents} of the based loop spaces \cref{BasedHigherLoopSpaces} form groups under concatenation, called the \emph{fundamental group} for $n = 1$ and generally the \emph{$n$th homotopy group}:
\begin{equation}
  \label{HomotopyGroups}
  \pi_n(Y, y_0)
    :=
  \pi_0\bracket({
    \Omega^n_{y_0} Y
  })
  \mathrlap{\,,}
\end{equation}
where the subscript (which we will suppress when understood) highlights that this depends on the chosen basepoint, in general. 

But in fact, up to isomorphism these groups depend only on the connected component $[y_0] \in \pi_0(Y)$ of the base point, since conjugate-concatenation with a path $\gamma$ connecting a pair of base points, $\gamma\bracket({0,\vec y}) = y_0$ and $\gamma\bracket({1,\vec y}) = y'_0$, gives a group isomorphism:
\begin{equation}
  \begin{tikzcd}
    \pi_n\bracket({
      Y, y'_0
    })
    \ar[
      rr,
      "{ \gamma \star (-) \star \overline{\gamma} }",
      "{ \sim }"{swap}
    ]
    &&
    \pi_n\bracket({
      Y, y'_0
    })
    \mathrlap{\,.}
  \end{tikzcd}
\end{equation}

Moreover, the homotopy groups \cref{HomotopyGroups} are functorial: Associated to pointed continuous maps are evident group homomorphisms, and this assignment preserves identities and composition (and sends homotopy equivalences \cref{HomotopyEquivalence} to group isomorphisms):
\begin{equation}
  \label{FunctorialityOfHomotopyGroups}
  \begin{tikzcd}[
    row sep=14pt,
    column sep=30pt
  ]
    & 
    Y
    \ar[dr, "{ g }"]
    \\
    X
    \ar[rr, "{ g \circ f }"]
    \ar[ur, "{ f }"]
    &&
    Z
  \end{tikzcd}
  \quad 
  \mapsto
  \quad 
  \begin{tikzcd}[
   row sep=14pt, 
   column sep=25pt
  ]
    &
    \pi_n\bracket(Y)
    \ar[
      dr,
      "{ \pi_n(g) }"{sloped}
    ]
    \\
    \pi_n(X)
    \ar[
      rr,
      "{ \pi_n( g \circ f) }"{sloped}
    ]
    \ar[
      ur,
      "{
        \pi_n(f)
      }"{sloped}
    ]
    &&
    \pi_n(Z)
    \mathrlap{\,.}
  \end{tikzcd}
\end{equation}

Next, a \emph{Serre fibration} is a continuous map $\inlinetikzcd{X \ar[r, "{ f }"] \& Y}$ such that all paths of $n$-paths in $Y$ may be lifted to paths of $n$-paths in $X$ for every choice of lift of the starting point:
\begin{equation}
  \label{SerreFibration}
  \text{
    $p$ is Serre fibration
  }
  \;\;\;\;
  \Leftrightarrow
  \;\;\;\;
  \forall_{n \in \mathbb{N}}
  \;\;
  \begin{tikzcd}
    \{0\} \times I^n
    \ar[
      d,
      hook
    ]
    \ar[
      r,
      "{ \forall }"
    ]
    &
    X
    \ar[d, "{ f }"]
    \\
    I \times I^n
    \ar[
      r,
      "{ \forall }"
    ]
    \ar[
      ur,
      dashed,
      "{ \exists }"
    ]
    &
    Y
    \mathrlap{\,.}
  \end{tikzcd}
\end{equation}
We write
\begin{equation}
  \label{FiberOfSerreFibration}
  \begin{tikzcd}[sep=20pt]
  \bracketmid\{{
    x \in X
  }{
    f(x) = y_0
  }\}
  =: 
  F
  \ar[r, hook, "{ \iota }"]
  &
  X
  \ar[r, "{ f }"]
  &
  Y
  \end{tikzcd}
\end{equation}
for the fiber space of a pointed Serre fibration \cref{SerreFibration}, 

A key example of a Serre fibration \cref{SerreFibration} is
\footnote{
  Under mild conditions, but it is sufficient that $X$ and $Y$ are CW-complexes.
}
the \emph{basepoint-evaluation map} on the mapping space \cref{MappingSpace} out of a pointed space, whose fiber is the pointed mapping space:
\begin{equation}
  \label{EvaluationSerreFibration}
  \begin{tikzcd}[
    row sep=-2pt
  ]
    \mathrm{Map}^\ast\bracket({
      X,Y
    })
    \ar[r, hook]
    &
    \mathrm{Map}\bracket({
      X,Y
    })
    \ar[
      r,
      "{ \mathrm{ev} }"
    ]
    &
    Y
    \\
    & 
    f \;\;
    \ar[
      r,
      |->,
      shorten=10pt
    ]
      &
    f(x_0)
    \mathrlap{\,.}
  \end{tikzcd}
\end{equation}

By functoriality \cref{FunctorialityOfHomotopyGroups}, a Serre fibration \cref{FiberOfSerreFibration} induces homotopy group homomorphisms of the form $\inlinetikzcd{ \pi_n(F) \ar[r] \& \pi_n(X) \ar[r] \& \pi_n(Y) }$. Moreover, there are \emph{connecting homomorphisms}
\begin{equation}
  \label{ConnectingHomomorphism}
  \begin{tikzcd}
    \pi_{n+1}(Y)
    \ar[r, "{ \partial_n }"]
    &
    \pi_n(X)
  \end{tikzcd}
\end{equation}
defined as follows:
Given $[\gamma] \in \pi_{n+1}(Y)$, we may choose 
\begin{enumerate}
\item
a representative 
$\inlinetikzcd{\gamma: I^{n+1} \ar[r] \& Y}$,
\item
a map $\inlinetikzcd{\widehat{\gamma}: I^{n+1} \ar[r] \& Y}$,

\item such that 
{\bf (a)}
  $\widehat \gamma_{
    \vert 
    (
      \partial I^{n+1} 
        \setminus
      \{1\} \times I^n
    )
   }
   = \mathrm{const}_{x_0}$
and {\bf (b)}
  $f \circ \widehat{\gamma} = \gamma$.
\end{enumerate}
Then 
\begin{equation}
  \partial_n\bracket({
    [\gamma]
  })
  :=
  \bracket[{\,
    \widehat{\gamma}_{\vert \{1\} \times I^n}
  }]
  \in
  \pi_n(F)
\end{equation}
is the well-defined value of the connecting homomorphism on $[\gamma]$.

The key fact now is:
\begin{proposition}[Homotopy long exact sequence (cf. {\cite[\S 9.8]{Fomenko2016}})]
\label[proposition]{HomotopyLES}
Given a Serre fibration \cref{SerreFibration}, the long sequence of homomorphisms of homotopy groups given by \cref{FunctorialityOfHomotopyGroups,ConnectingHomomorphism} 
\begin{equation}
  \begin{tikzcd}[
    column sep=50pt
  ]
    {}
    \ar[r, - , dotted]
    &
    \pi_{n+1}(F)
    \ar[r, "{ \pi_{n+1}(\iota) }"]
    &
    \pi_{n+1}(X)
    \ar[r, "{ \pi_{n+1}(f) }"]
    &
    \pi_{n+1}(Y)
    \ar[
      dll,
      snake left,
      "{ \partial_n }"{description}
    ]
    \\
    &
    \pi_{n}(F)
    \ar[r, "{ \pi_{n+1}(\iota) }"]
    &
    \pi_{n}(X)
    \ar[r, "{ \pi_{n+1}(f) }"]
    &
    \pi_{n}(Y)
    \ar[r, -,  dotted]
    &
    {}
  \end{tikzcd}
\end{equation}
is exact, meaning that at each vertex the image of the incident map coincides with the kernel of the outgoing map. This continues to hold for the maps of pointed sets in degree 0, 
\begin{equation}
  \label{TheHomotopyLESInLowDegree}
  \begin{tikzcd}[column sep=35pt]
  &
  {}
  \ar[
    r,
    -,
    dotted
  ]
  & \pi_1(Y)
  \ar[
    dll,
    snake left,
    "{ \partial_0 }"{description}
  ]
  \\
  \pi_0(F)
  \ar[r, "{ \pi_0(\iota) }"]
  &
  \pi_0(X)
  \ar[r, "{ \pi_0(f)}"]
  &
  \pi_0(Y)
  \mathrlap{\,,}
  \end{tikzcd}
\end{equation}
where kernels are understood as preimages of the base point.
\end{proposition}

\subsection{Groupoids}

Introduction to groupoids and their representation theory may be found in \cite{IbortRodriguez2021}.

\subsubsection{Groupoid Representations}
\label{OnGroupoidRepresentations}

For $\mathcal{G}$ a groupoid (regarded as small category), a (finite-dimensional, complex) \emph{representation} is simply a functor to the category $\mathrm{Vec}$, the category of finite-dimensional complex vector spaces, and a homomorphism/intertwiner of representations is a natural transformations between these functors. Hence the category of $\mathcal{G}$-representations is just the functor category:
\begin{equation}
  \label{CategoryOfGroupoidRepresentations}
  \mathrm{Rep}(\mathcal{G})
  :=
  \mathrm{Func}\bracket({
    \mathcal{G},
    \mathrm{Vec}
  })
  \mathrlap{\,.}
\end{equation}
For instance, for $G$ a group, with \emph{delooping groupoid}
\begin{equation}
  \label{DeloopingGroupoid}
  \mathbf{B}G
  :=
  \Big\{\!\!\!\!
  \adjustbox{
    raise=-7pt
  }{
  \begin{tikzcd}
    \bullet
    \ar[
      in=50,
      out=180-50,
      looseness=6,
      shift right=4pt,
      "{ 
        \,\mathclap{g}\, 
      }"{description}
    ]
  \end{tikzcd}
  }\!\!
  \Big\vert
   \,
    g \in G
  \Big\}
  \mathrlap{\,,}
\end{equation}
the category of ordinary $G$-representations is equivalently that of groupoid representations of $\mathbf{B}G$.

In fact, for $\pi_0(\mathcal{G})$ the set of connected components, and $\pi_1\bracket({\mathcal{G}, c})$ the isotropy group of an object $c \in \mathrm{Obj}(\mathcal{G})$, the following inclusion (of the delooping groupoids \eqref{DeloopingGroupoid} of the isotropy groups on representative objects in each connected component) is an equivalence
\begin{equation}
  \label{GroupoidEquivalentToDisjointUnionOfDeloopings}
  \begin{tikzcd}[
    row sep=10pt, column sep=large
  ]
  \displaystyle{
  \bigsqcup_{
      [c] \in 
      \pi_0(\mathcal{G})
  }
  }
  \mathbf{B}\bracket({
    \pi_1\bracket({\mathcal{G}, c})
  })
  \ar[
    r,
    hook,
    "{ \iota }",
    "{ \sim }"{swap}
  ]
  &
  \mathcal{G}
  \\[-20pt]
  \;\;\;\;\;\;\bullet_{[c]}
  \ar[
    r,
    |->,
    shorten=5pt
  ]
  &
  c
  \mathrlap{\,,}
  \end{tikzcd}
\end{equation}
so that groupoid representations are equivalently indexed tuples of isotropy group representations:
\begin{equation}
  \label{GroupoidRepAsTupleOfGroupReps}
  \begin{tikzcd}
  \mathrm{Rep}(\mathcal{G})
  \ar[
    r,
    "{ \iota^\ast }",
    "{ \sim }"{swap}
  ]
  &
  \displaystyle{
  \prod_{
      [c] \in 
      \pi_0(\mathcal{G})
  }
  }
  \mathrm{Rep}\bracket({
    \pi_1\bracket({\mathcal{G}, c})
  })
  \mathrlap{\,.}
  \end{tikzcd}
\end{equation}
Notice that for any other choice of representative $c'$ in the connected component $[c]$ we find a morphism $\inlinetikzcd{c \ar[r, "{\gamma_c}"] \& c'  }$ in $\mathcal{G}$.
This induces a group isomorphism along which representations pull back:
\begin{equation}
  \begin{tikzcd}[
    row sep=0pt
  ]
    \pi_1\bracket({
      \mathcal{G},c'
    })
    \ar[
      rr,
      "{
        \gamma_c^{-1}
          \circ 
        (-)
          \circ 
        \gamma_c
      }"
    ]
    &&
    \pi_1\bracket({
      \mathcal{G},c
    })
    \\
    \mathrm{Rep}\bracket({
      \pi_1\bracket({\mathcal{G},c'})
    })
    \ar[
      <-,
      rr,
      "{ (-)^{\gamma_c} }",
      "{ \sim }"{swap}
    ]
    &&
    \mathrm{Rep}\bracket({
      \pi_1\bracket({\mathcal{G},c})
    })
    \mathrlap{\,,}
  \end{tikzcd}
\end{equation}
whence the above equivalence \eqref{GroupoidRepAsTupleOfGroupReps} transforms as follows under change of base points along such choices of morphisms $\gamma_c$:
\begin{equation}
  \label{IsotropyRepsTransformWithBasePoints}
  \begin{tikzcd}[
    row sep=2pt, column sep=large
  ]
  &
  \prod_{
      [c] 
  }
  \mathrm{Rep}\bracket({
    \pi_1\bracket({\mathcal{G}, c})
  })
  \ar[
    dd,
    "{
      \prod_{[c]}
      (-)^{\gamma_c}
    }"
  ]
  \\
  \mathrm{Rep}(\mathcal{G})
  \ar[
    ur,
    "{ \iota^\ast }",
    "{ \sim }"{swap, sloped}
  ]
  \ar[
    dr,
    "{ \iota'^\ast }"{swap},
    "{ \sim }"{sloped}
  ]
  \\
  &
  \prod_{
      [c'] 
  }
  \mathrm{Rep}\bracket({
    \pi_1\bracket({\mathcal{G}, c'})
  })  
  \mathrlap{\,.}
  \end{tikzcd}
\end{equation}

\subsubsection{Inertia Groupoids}

For example, for $G \acts S$ a group action on a set, the corresponding \emph{action groupoid} is
\begin{equation}
  \label{ActionGroupoid}
  S \sslash G
  \defneq
  \left\{
  \begin{tikzcd}[
    ampersand replacement=\&, row sep=small,
    column sep=10pt
  ]
    \&
    g_1 \!\cdot\! s
    \ar[
      dr,
      "{ g_2 }"
    ]
    \\
    s
    \ar[
      ur,
      "{ g_1 }"
    ]
    \ar[
      rr,
      "{ g_2 g_1 }"
    ]
    \&\&
    g_2 g_1 \!\cdot\! s
  \end{tikzcd}
  \,\middle\vert\,
    \substack{
      s \in S
      \\
      g_i \in G
    }
  \right\}
  \,,
\end{equation}
its connected components are the orbits of the action, and the isotropy group of $s \in S$ is the \emph{stabilizer group}
\begin{equation}
  \label{IsotropyGroup}
  \pi_1\bracket({
    S \sslash G
    ,s
  })
  \simeq
  \bracketmid\{{
    g \in G
  }{
    g \cdot s = s
  }\}
  \mathrlap{\,.}
\end{equation}
Specifically, for the adjoint action $G \acts_{\mathrm{Ad}}\, G$ of the group on itself, the action groupoid \eqref{ActionGroupoid} is also known as the \emph{inertia groupoid}:
\begin{equation}
  \label{InertiaGroupoid}
  \Lambda G
  :=
  G \sslash_{\!\mathrm{Ad}} G
  =
  \left\{
  \begin{tikzcd}[row sep=small,
    column sep=10pt]
    &
    g^{k_1}
    \ar[
      dr,
      shorten <=-4pt,
      shorten >=-3pt,
      "{ k_2 }"
    ]
    \\
    g
    \ar[
      ur,
      "{ k_1 }"
    ]
    \ar[
      rr,
      "{ k_2 k_1 }"
    ]
    &&
    g^{k_2 k_1}
  \end{tikzcd}
  \,\middle\vert\,
    k_i, g \in G
  \right\}
  \mathrlap{\,,}
\end{equation}
where we denote conjugation by
\begin{equation}
  g^k := k g k^{-1}
  \mathrlap{.}
\end{equation}
Its isotropy groups \eqref{IsotropyGroup} are the \emph{centralizer subgroups}:
\begin{equation}
  \label{CentralizersAsIsotropyGroups}
  \pi_1\bracket({
    \Lambda G
    ,
    g
  })
  \simeq
  Z_G(g)
  \mathrlap{\,.}
\end{equation}
Notice that we have group isomorphisms
$$
  \begin{tikzcd}[row sep=-3pt,
    column sep=0pt
  ]
    Z_G\bracket({g^k})
    \ar[
      rr,
      "{  }",
      "{ \sim }"{swap}
    ]
    &&
    Z_G({g})
    \\
    q &\longmapsto& q^{(k^{-1})}
  \end{tikzcd}
$$
along which we may pull back group representations:
\begin{equation}
  \label{RepsTransformUnderGroupConjugation}
  \begin{tikzcd}[
    column sep=large]
    \mathrm{Rep}\bracket({
      Z_G(g)
    })
    \ar[
      r,
      "{ (-)^k }",
      "{ \sim }"{swap}
    ]
    &
    \mathrm{Rep}\bracket({
      Z_G\bracket({g^k})
    })
    \mathrlap{\,.}
  \end{tikzcd}
\end{equation}

\subsubsection{Induced Representations}
\label{OnInducedRepresentations}

For $\inlinetikzcd{H \ar[r, hook, "{ \iota }"] \& G}$ a subgroup inclusion, the delooping groupoid $\mathbf{B}H$ \eqref{DeloopingGroupoid} is equivalent to the action groupoid of $G$ acting by left multiplication on the cosets $G/G$
$$
    \mathbf{B}H
    \simeq
    G \backsslash G / H\,.
$$
This resolves $\mathbf{B} \iota$ by a Kan fibration:
$$
  \begin{tikzcd}
    \mathbf{B}H
    \ar[
      rrr,
      uphordown,
      "{ \mathbf{B} \iota }"
    ]
    \ar[
      r,
      "{ e }",
      "{ \sim }"{swap}
    ]
    &
    G \backsslash G/H
    \ar[
      rr,
      "{ p }",
      "{
        \text{fib}
      }"{swap}
    ]
    &&
    \mathbf{B}G
    \mathrlap{\,.}
  \end{tikzcd}
$$
Along the Kan fibration $p$, pushforward $p_!$ is given simply by direct sum of fibers. Under the equivalence $e$ this coincides with forming left induced representations. 

\subsection{Drinfeld Center}
\label{OnTheDrinfeldCenter}

 We recall the \emph{Drinfeld center} $\mathcal{Z}(G) := \mathcal{Z}\bracket({\mathrm{Vec}_G})$ of the category of $G$-graded vector spaces, and the nature of its simple objects (which is classical, cf. \cite[Ex. 8.5.4]{EtingofEtAl2015}), making explicit that it is equivalently the representation category of the adjoint action groupoid (\cref{DrinfeldCenterAsInertiaGroupoidRepresentations}). 
 Then we derive its fusion rules (\cref{FusionCoefficientsInDrinfeldCenterOfGVectorSpaces}, which seem not to be readily citable from the literature; in dual Hopf-algebraic formulation it was only recently claimed in \cite[Thm. 3.2]{Li2026}).

\subsubsection{The Fusion Category}

We begin with recalling relevant definitions and establishing notation (cf. \cite{EtingofEtAl2015} for background on braided tensor categories). 

By $\mathrm{Vec} := \bracket({\mathrm{FD}\mathbb{C}\mathrm{Vec}, \otimes, \mathbb{C}})$ we denote the monoidal category of finite-dimensional complex vector spaces with the usual tensor product. 
For $G$ a group, the monoidal category $\mathrm{Vec}_G$ of \emph{$G$-graded} vector spaces has (cf. \cite[Ex. 2.3.6]{EtingofEtAl2015}):
\begin{enumerate}

\item
as \emph{objects} vector spaces $V \in \mathrm{Vec}$ equipped with $G$-indexed direct sum decomposition
\begin{equation}
  V
  \defneq
  \bigoplus_{g \in G}
  V_g
  \mathrlap{\,,}
\end{equation}

\item
as \emph{morphisms} linear maps that respect the $G$-grading,

\item
as \emph{tensor product} $\fusion$ the ordinary tensor product of underlying vector spaces equipped with the convolution grading:
\begin{equation}
  \label{TensorProductOnGradedVectorSpaces}
  \bracket({V \fusion W})_g
  :=
  \bigoplus_{k \in G}
  \bracket({
    V_{k} \otimes W_{k^{-1} g}
  })
  \mathrlap{\,.}
\end{equation}
We shall refer to $\fusion$ in  \eqref{TensorProductOnGradedVectorSpaces}, and its incarnation \eqref{TensorProductInDrinfeldCenter} in the Drinfeld center below, as the \emph{fusion product}, in order to distinguish it from the degreewise tensor product (cf. Rem. \ref{EquivalenceIsNotMonoidal} below). 
\end{enumerate}

Equivalently, with 
\begin{equation}
  \label{HopfAlgebraOfFunctionsOnG}
  C(G) := \mathrm{Map}\bracket({G,\mathbb{C}})
\end{equation}
denoting the Hopf algebra of complex-valued functions on the underlying set of $G$ (with product the pointwise product of functions and with coproduct and antipode induced by precomposition of functions with the group multiplication and group inverse, respectively), we have that $\mathrm{Vec}_G$ is equivalent (as a monoidal category) to the category of modules over $C(G)$:
\begin{equation}
  \mathrm{Vec}_G 
    \simeq 
  \mathrm{Mod}\bracket({C(G)})
  \mathrlap{\,.}
\end{equation}
(The pointwise product in $C(G)$ induces the $G$-grading and the coproduct induces the degree convolution in \eqref{TensorProductOnGradedVectorSpaces}.)

For example, with $X \in \mathrm{Vec}$ an ungraded vector space and $g_0 \in G$ a group element, we write
\begin{equation}
  \label{GradedVectorSpaceConcentratedInSingleDegree}
  X \delta_{g_0}
  \in 
  \mathrm{Vec}_G
  \,,
  \;\;\;
  \bracket({X \delta_{g_0}})_g
  :=
  \begin{cases}
    X & \text{if $g = g_0$}
    \\
    0 & \text{otherwise}
  \end{cases}
\end{equation}
for the $G$-graded vector space which is concentrated on $X$ in degree $g_0$.

Now, the \emph{Drinfeld center} $\mathcal{Z}\bracket({\mathrm{Vec}_G})$ is, by its general definition (cf. \cite[\S 7.13]{EtingofEtAl2015}), the braided monoidal category which has:
\begin{enumerate}

\item
  as \emph{objects} pairs, $(V,\beta)$, consisting of a $V \in \mathrm{Vec}_G$
  and a natural isomorphism (the \emph{half-braiding})
  \begin{equation}
    \label{TheHalfBraiding}
    \begin{tikzcd}
    \beta_{(-)}
    :
    (-) \widetilde\otimes V
    \ar[r, "{ \sim }"]
    &
    V \widetilde\otimes (-)
    \mathrlap{\,,}
    \end{tikzcd}
  \end{equation}
  satisfying for all $W, W' \in \mathrm{Vec}_G$ the relation
  \begin{equation}
    \label{CoherenceOfHalfBraiding}
    \beta_{W \fusion W'}
    =
    \bracket({
      \beta_{W} 
        \fusion 
      \mathrm{id}_{W'}
    })
      \circ
    \bracket({
      \mathrm{id}_{W} 
        \fusion 
      \beta_{W'}
    })
    \mathrlap{\,,}
  \end{equation}

  \item
  as \emph{morphisms} 
  $\inlinetikzcd{(V,\beta) \ar[r] \& (V',\beta')}$
  linear maps
  $\inlinetikzcd{f : V \ar[r] \& V'}$ such that for 
  all $W \in \mathrm{Vec}_G$ the following diagram commutes (cf. \cite[(7.41)]{EtingofEtAl2015}):
  \footnote{
    We are notationally suppressing the associators. 
  }
  \begin{equation}
    \label{ConditionOnMorphismsInDrinfeldCenter}
    \begin{tikzcd}[
      column sep=50pt
    ]
      W \fusion V
      \ar[
        d,
        "{ \beta_W }"
      ]
      \ar[
        r,
        "{
          \mathrm{id}_W
            \fusion 
          f 
        }"
      ]
      &
      W \fusion V'
      \ar[
        d,
        "{ \beta'_W }"
      ]
      \\
      V \fusion W
      \ar[
        r,
        "{
          f
          \fusion
          \mathrm{id}_W
        }"
      ]
      &
      V' \fusion W
      \mathrlap{\,;}
    \end{tikzcd}
  \end{equation}
  \item
  as \emph{tensor product} the operation which on the underlying vector spaces is \eqref{TensorProductOnGradedVectorSpaces} and extended to the half-braiding as:
  \begin{equation}
    \label{TensorProductInDrinfeldCenter}
    (V,\beta)
    \,\fusion \,
    (V',\beta')
    =
    \bracket({
      V \fusion V',
      (\mathrm{id}_{V} \fusion \beta')
      \circ
      (\beta \fusion \mathrm{id}_{V'})
    })
    \mathrlap{\,,}
  \end{equation}

   \item
   as \emph{braiding} (cf. \cite[(8.15)]{EtingofEtAl2015}) on a pair of objects the half-braiding \cref{TheHalfBraiding} of the second object around the first:
   \begin{equation}
     \label{BraidingInDrinfeldCenter}
     \begin{tikzcd}
       b_{
         (V,\beta),
         (V',\beta')
       }
       :=
       \beta'_{V}       
      \, : \,
       V \fusion V'
       \ar[r, "{ \sim }"]
       &
       V' \fusion V
       \mathrlap{\,.}
     \end{tikzcd}
   \end{equation}
\end{enumerate}

\begin{remark}[{\cite[\S 7.14]{EtingofEtAl2015}}]
\label{TheDrinfeldDouble}
This Drinfeld center $\mathcal{Z}\bracket({\mathrm{Vec_G}})$ is equivalent, as a monoidal category, to the category of modules over the \emph{Drinfeld double} Hopf algebra $D\bracket({C\bracket({G})})$ of \eqref{HopfAlgebraOfFunctionsOnG}:
\begin{equation}
  \mathcal{Z}\bracket({
    \mathrm{Vec}_G
  })
  \simeq
  \mathrm{Mod}\bracket({
    D\bracket({
      C\bracket({G})
    })
  })
  \mathrlap{\,.}
\end{equation}
Much of the literature discusses the Drinfeld center in this dual Hopf algebraic form.
\end{remark}

\subsubsection{Simple Objects via Inertia Irreps}

The above general definition of the Drinfeld center applied to the case of $\mathrm{Vec}_G$ may be simplified using the notion of groupoid representations from \cref{OnGroupoidRepresentations}: 

Since every vector space is isomorphic to a direct sum of copies of $\mathbb{C}$, the linear maps $\beta_{(-)}$ \eqref{TheHalfBraiding} are determined already by their values on $\mathbb{C} \delta_{g_0}$ \eqref{GradedVectorSpaceConcentratedInSingleDegree}; and if we understand that $\mathbb{C} \fusion V = V$, then, by degree reasons, their $g_0 g$-components must be linear isomorphisms $g_0 \cdot (-)$ of this form:
\begin{equation}
  \begin{tikzcd}[
    column sep=50pt,
    row sep=10pt
  ]
    \bracket({
      \mathbb{C} \delta_{g_0}
        \fusion 
      V 
    })_{g_0 g}
    \ar[
      r,
      "{ 
        ({
          \beta_{\mathbb{C}\delta_{g_0}} 
        })_{g_0 g}
      }"
    ]
    \ar[d, equals]
    &
    \bracket({
      V
        \fusion 
      \mathbb{C}\delta_{g_0} 
    })_{g_0 g}
    \ar[d, equals]
    \\
    V_{g} 
    \ar[r, "{ g_0\cdot(-) }"]
      &
    V_{\mathrlap{g_0 g g_0^{-1}.}}
  \end{tikzcd}
\end{equation}
For these, the coherence condition \eqref{CoherenceOfHalfBraiding} reduces to the \emph{action property}
\begin{equation}
  g_2 \cdot \bracket({g_1 \cdot (-)})
  =
  (g_2 g_1) \cdot (-)
  \mathrlap{\,.}
\end{equation}
This way, the objects of $\mathcal{Z}\bracket({\mathrm{Vec}_G})$ are bijectively identified with (complex, finite-dimensional) \emph{groupoid representations} of the inertia groupoid $\Lambda G$ \eqref{InertiaGroupoid}:
\begin{equation}
  (C,\beta)
  \qquad 
  \leftrightarrow
  \qquad
  \begin{tikzcd}[
    row sep=0pt,
    column sep=0pt
  ]
    \Lambda G
    \ar[
      rr
    ]
    &&
    \mathrm{Vec}
    \\
    g
    \ar[
      d,
      "{\, g_1 }"
    ]
    &\longmapsto&
    V_{\mathrlap{g}}
    \ar[
      d,
      "{\, g_1 \cdot (-) }"
    ]
    \\[13pt]
    g_1
    g 
    g_1^{-1}
    &\longmapsto&
    V_{\mathrlap{g_1 g g_1^{-1}.}}
  \end{tikzcd}
\end{equation}
In this representation-theoretic description, the morphisms \eqref{ConditionOnMorphismsInDrinfeldCenter} of $\mathcal{Z}(\mathrm{Vec}_G)$ are equivalently the homomorphisms of groupoid representations (the \emph{intertwiners}). This shows that:
\begin{lemma}
  \label[lemma]{DrinfeldCenterAsInertiaGroupoidRepresentations}
  The Drinfeld center of $\mathrm{Vec}_G$ is equivalently the representation category \eqref{CategoryOfGroupoidRepresentations} of the inertia groupoid of $G$:
  \begin{equation}
    \label{DrinfeldCenterEquivalentToInertiaGroupoidReps}
    \mathcal{Z}\bracket({\mathrm{Vec}_G})
    \simeq
    \mathrm{Rep}\bracket({
      \Lambda G
    })
    \mathrlap{\,.}
  \end{equation}
\end{lemma}

\cref{DrinfeldCenterAsInertiaGroupoidRepresentations} implies at once that \emph{simple objects} in $\mathcal{Z}\bracket({\mathrm{Vec}_G})$ (those admitting no nontrivial direct sum decomposition) are supported on conjugacy classes of group elements
\begin{equation}
  \label{ConjugacyClass}
  [g]
  :=
  \bracketmid\{{
    g^k
      =
    k g k^{-1}
  }{
    k \in G
  }\}
  \mathrlap{\,,}
\end{equation}
where they form an (irreducible) groupoid representation (cf. \cref{OnGroupoidRepresentations}) of the connected component $[g] \sslash_{\!\mathrm{Ad}} G \subset G\sslash_{\!\mathrm{Ad}} G$. Since these connected groupoids are equivalent \eqref{GroupoidEquivalentToDisjointUnionOfDeloopings} to the delooping of their isotropy group, $[g] \sslash_{\!\mathrm{Ad}} G \simeq \mathbf{B}Z_G(g)$, these groupoid representations are equivalently \eqref{GroupoidRepAsTupleOfGroupReps} representations of the centralizer group $Z_G(g)$ \eqref{CentralizersAsIsotropyGroups}:

\begin{proposition}
\label[proposition]{SimpleObjects}
The simple objects of $\mathcal{Z}\bracket({\mathrm{Vec}_G})$ are pairs $\bracket({[g],\rho})$ consisting of a conjugacy class $[g]$ and an irrep of $Z_G(g)$:
\begin{equation}
  \label{TheSimpleObjects}
  \Big\{
    \text{\rm simple objects}
  \Big\}
  \,
  \simeq
  \bigsqcup_{
    \!\!\!
    [g] \in
    \mathrm{Conj}(G)
    \!\!\!
  }
  \mathrm{Irr}\bracket({
    [g]\sslash_{\!\mathrm{Ad}}
    G
  })
  \;
  \simeq
  \bigsqcup_{
    \!\!\!
    [g] \in
    \mathrm{Conj}(G)
    \!\!\!
  }
  \mathrm{Irr}\bracket({
    Z_G(g)
  })
  \mathrlap{\,.}
\end{equation}
\end{proposition}

\subsubsection{The Fusion Product}
\label{OnTheFusionProduct}

\begin{remark}
\label{EquivalenceIsNotMonoidal}
Beware that the above equivalence \eqref{DrinfeldCenterEquivalentToInertiaGroupoidReps} is \emph{not} a monoidal equivalence, as the standard tensor product on groupoid representations is degree-wise, not convolutive as in \eqref{TensorProductOnGradedVectorSpaces}.
\end{remark}

Instead, the fusion product in the Drinfeld center is category-theoretically given as follows:
\begin{lemma}
  \label[lemma]{FusionViaPullPush}
  For a pair of objects $F_1,F_2 \in \mathrm{Rep}(\Lambda G) \simeq \mathcal{Z}\bracket({\mathrm{Vec}_G})$
  \eqref{DrinfeldCenterEquivalentToInertiaGroupoidReps}, their fusion tensor product $\fusion$ \eqref{TensorProductInDrinfeldCenter} is equivalently given by the following pull-tensor-push operation:
  \begin{equation}
    \begin{tikzcd}[
      column sep=40pt,
      row sep=0pt
    ]
      \bracket({
        G\sslash_{\!\mathrm{Ad}}G
      })
      \times
      \bracket({
        G\sslash_{\!\mathrm{Ad}}G
      })
      \ar[
        r,
        <-,
        "{
          (\mathrm{pr}_1,\mathrm{pr}_2)
        }"
      ]
      &
      \bracket({G \times G})
      \sslash_{\!\mathrm{Ad}}G
      \ar[
        r,
        "{ 
          m :=
          (-)\cdot(-) 
        }"
      ]
      &[-15pt]
      G\sslash_{\!\mathrm{Ad}}G
      \\
      \bracket({F,F'})
      \ar[
        rr,
        |->,
        shorten=10pt
      ]
      &&
      m_!\bracket({
        \bracket({\mathrm{pr}_1^\ast F})
        \otimes
        \bracket({\mathrm{pr}_2^\ast F})
      })
      \mathrlap{\,,}
    \end{tikzcd}
  \end{equation}
  in that there is a natural isomorphism
  \begin{equation}
    F_1 \fusion F_2
    \simeq
      m_!\bracket({
        \bracket({\mathrm{pr}_1^\ast F})
        \otimes
        \bracket({\mathrm{pr}_2^\ast F})
      })    
    \mathrlap{\,.}
  \end{equation}
\end{lemma}
\begin{proof}
  Since the functor $m$ is a Kan fibration, the left base change $m_!$ is given by forming the direct sum over fibers. (It is immediate to check the universal property of the left adjoint explicitly.) This reproduces the definition \eqref{TensorProductOnGradedVectorSpaces}. The induced $G$-action on these direct sums is the diagonal one, which under the equivalence \eqref{DrinfeldCenterEquivalentToInertiaGroupoidReps} reproduces the definition \eqref{TensorProductInDrinfeldCenter}. 
\end{proof}

\begin{lemma}
  \label[lemma]{OnTheDecompositionOfProductOrbits}
  For $g_1,g_2 \in G$, with $[g_1] \times [g_2] \subset G^2$ denoting the Cartesian product of their conjugacy classes \eqref{ConjugacyClass}, the action groupoid \cref{ActionGroupoid} of the diagonal adjoint action is equivalent to the following disjoint union of delooping groupoids \cref{DeloopingGroupoid}:
\begin{equation}
  \label{DecompositionProductOrbits}
  \bracket({
    [g_1]
      \times
    [g_2]
  })
  \sslash_{\!\mathrm{Ad}} G
  \;\;
  \simeq
  \bigsqcup_{
    \!\!
    \substack{
      [g] \in
      \mathrm{Conj}(G)
      \\
      [k_1,k_2] 
      \in 
      \FusionMultiplicity_{g_1, g_2}^g
    }
    \!\!
  }
  \mathbf{B}\bracket({
    Z_G\bracket({
      g_1^{k_1}, g_2^{k_2}
    }) 
  })
  \mathrlap{\,,}
\end{equation}
where the direct sum on the right is indexed over the set \cref{DoubleCosetInProductOrbitDecomposition}:
$$
  \FusionMultiplicity_{g_1,g_2}^g
  :=
  \bracketmid\{{
    [k_1,k_2] \in
    Z_G(g) 
    \big\backslash  
    G^2 
     \big/ 
    \bracket({Z_G(g_1) \times Z_G(g_2)})
  }{
    g_1^{k_1} g_2^{k_2} = g
  }\}
  \mathrlap{\,.}
$$
\end{lemma}
\begin{proof}
By \cref{TheFusionMultiplicity}, the disjoint union on the right of \eqref{DecompositionProductOrbits} is indexed by the set of orbits of the groupoid on the left. Since $Z_G\bracket({g_1^{k_1}, g_2^{k_2}})$ is clearly the isotropy group of the given connected component, the claim follows by \eqref{GroupoidEquivalentToDisjointUnionOfDeloopings}.
\end{proof}

We now have the following formula for the \emph{fusion coefficients} 
$$
  N
    ^{([g],\rho)}
    _{ ([g_1],\rho_1), ([g_2],\rho_2)}
  \in
  \mathbb{N}
  \mathrlap{\,,}
$$
which count the multiplicity of simple objects appearing in the fusion product of pairs of simple objects:
\begin{equation}
  \label{FusionRule}
  \bracket({
    [g_1], \rho_1
  })
\,  \fusion \,
  \bracket({
    [g_2], \rho_2
  })
  \simeq
  \bigoplus_{
    \scaledbracket({[g],\rho})
  }
  N
    ^{([g],\rho)}
    _{ ([g_1],\rho_1), ([g_2],\rho_2)}
  \,
  \bracket({[g],\rho})
  \mathrlap{\,.}
\end{equation}

\begin{proposition}
\label[proposition]{FusionCoefficientsInDrinfeldCenterOfGVectorSpaces}
The fusion coefficients \eqref{FusionRule} in $\mathcal{Z}\bracket({\mathrm{Vec}_G})$
are given by:
\begin{equation}
  \label{TheFusionCoefficientsInDrinfeldCenterOfGVectorSpaces}
  N
    ^{([g],\rho)}
    _{ ([g_1],\rho_1), ([g_2],\rho_2)}
  =
  \sum_{
    [k_1,k_2] 
      \in 
    \FusionMultiplicity_{g_1, g_2}^g
  }
  \!\!\!
  \mathrm{dim}\,
  \mathrm{Hom}_{
    Z_G\bracket({\scriptstyle
      g_1^{k_1}, g_2^{k_2}
    })
  }
  \bracket({
    \rho_1^{k_1}
    \otimes
    \rho_2^{k_2},
    \rho
  })
  \mathrlap{\,,}
\end{equation}
where $
  \FusionMultiplicity_{g_1,g_2}^g
  :=
  \bracketmid\{{
    [k_1,k_2] \in
    Z_G(g) 
    \big\backslash  
    G^2 
     \big/ 
    \bracket({Z_G(g_1) \times Z_G(g_2)})
  }{
    g_1^{k_1} g_2^{k_2} = g
  }\}
$ \eqref{DoubleCosetInProductOrbitDecomposition}.
\end{proposition}
\begin{proof}
  By Lem. \ref{FusionViaPullPush}, the fusion product of simple objects \eqref{TheSimpleObjects} is given by pull-push through
  $$
    \begin{tikzcd}[
      column sep=25pt,
      row sep=0pt
    ]
      \bracket({
        [g]\sslash_{\!\mathrm{Ad}}G
      })
      \times
      \bracket({
        [g']\sslash_{\!\mathrm{Ad}}G
      })
      \ar[
        <-,
        rr,
        "{
          (\mathrm{pr}_1,\mathrm{pr}_2)
        }"
      ]
      &&
      \bracket({
        [g]\times [g']
      })\sslash_{\!\mathrm{Ad}} 
      G
      \ar[
        r,
        "{ m }"
      ]
      &
      G \sslash_{\!\mathrm{Ad}} G
      \mathrlap{\,.}
    \end{tikzcd}
  $$
  
By Lemma \ref{OnTheDecompositionOfProductOrbits} and with \eqref{IsotropyRepsTransformWithBasePoints} and \eqref{RepsTransformUnderGroupConjugation}, this yields on the orbit $[g]$ the isotropy representation
\begin{equation}
  \bigoplus_{
    [k_1, k_2] 
      \in 
    R_{g_1, g_2}^g
  }
  \!\!\!
  \mathrm{Ind}
    _{Z_G(g_1^{k_1}, g_2^{k_2})}
    ^{Z_G(g)}
  \bracket({
    \mathrm{Res}
      ^{Z_G(g_1^{k_1}, g_2^{k_2})}
      _{Z_G(g)}
    \bracket({
      \rho_1^{k_1}
    })
    \otimes
    \mathrm{Res}
      ^{Z_G(g_1^{k_1}, g_2^{k_2})}
      _{Z_G(g)}
    \bracket({
      \rho_2^{k_2}
    })
  })
  \mathrlap{\,,}
\end{equation}
where we are using, with \cref{OnInducedRepresentations}, that on group representations the left base change $m_!$ is given by forming induced representations.

This yields the claim: By Schur's lemma, the multiplicity of the irrep $\rho$ in this expression is the dimension of the hom-space from the latter to the former. Finally, by Frobenius reciprocity the induction is left adjoint to restriction, whence we have \eqref{TheFusionCoefficientsInDrinfeldCenterOfGVectorSpaces}.
\end{proof}

\begin{remark}
  \label[remark]{LiteratureOnDrinfeldFusionProduct}
  Up to immediate translation, \cref{FusionCoefficientsInDrinfeldCenterOfGVectorSpaces} coincides with \cite[Thm. 3.2]{Li2026}, proven there by different methods (Mackey theory applied to the Drinfeld double, cf. Rem. \ref{TheDrinfeldDouble}). As remarked there, this is in turn the untwisted specialization of a formula given by \cite[Thm. 4.5]{Goff2012}.
\end{remark}

\medskip 
\subsection*{Acknowledgments}
We thank 
Alonso Perez-Lona 
and 
Sachin Valera
for useful discussion. 

\medskip 
\printbibliography

\end{document}